\documentclass[final,9pt]{IEEEtran}

\usepackage{graphicx}

\usepackage{amsfonts}
\usepackage{amssymb}
\usepackage{amsmath}
\usepackage{float}
\usepackage{color}
\usepackage{amscd}                 
\usepackage{pxfonts}

\newcommand{\ul}{\underline}

\newcommand{\mc}{\mathcal}

\newcommand{\ds}{\displaystyle}

\newtheorem{theorem}{Theorem}[section]
\newtheorem{definition}[theorem]{Definition}

\newtheorem{proposition}[theorem]{Proposition}

\usepackage{algorithm}
\usepackage{algpseudocode}

\begin{document}

\title{A Cross-Layer Approach for Distributed Energy-Efficient Power Control in Interference Networks}
 \author{Vineeth S. Varma$^{1,2}$, Samson Lasaulce$^1$, Yezekael Hayel$^3$, and Salah Eddine Elayoubi$^2$
 \thanks{$^1$ CNRS-Supelec-Univ. Paris Sud 11, 91192 Gif-sur-Yvette, France, email: \{vineeth.varma,samson.lasaulce\}@lss.supelec.fr}
\thanks{$^2$ Orange Labs, 92130 Issy-les-Moulineaux, France, email: salaheddine.elayoubi@orange.com}
\thanks{$^3$ University of Avignon, 84911 Avignon, France, email: yezekael.hayel@univ-avignon.fr}

}%

\markboth{IEEE Transactions on Vehicular Technology}%
{Trans}%
\maketitle

\begin{abstract}
In contrast with existing works which rely on the same type of energy-efficiency measure to design distributed power control policies, the present work takes into account the presence of a finite packet buffer at the transmitter side and the impact of transport protocols. This approach is relevant when the transmitters have a non-zero energy cost even when the radiated power is zero. A generalized energy-efficiency performance metric integrating these features is constructed under two different scenarios in terms of transport layer protocols characterized by a constant or an adaptive packet arrival rate. The derived performance metric is shown to have several attractive properties in both scenarios, which ensures convergence of the used distributed power control algorithm to a unique point. This point is the Nash equilibrium of a game for which the equilibrium analysis is conducted. Although the equilibrium analysis methodology is not new in itself, conducting it requires several non-trivial proofs, including the proof of quasi-concavity of the payoff functions. A thorough numerical analysis is provided to illustrate the effects of the proposed approach, and provides several valuable insights in terms of designing interference management policies.
\end{abstract}

\begin{IEEEkeywords}
Cross-layer, distributed optimization, distributed power control, energy-efficiency, game theory, Nash equilibrium, non-cooperative game.
\end{IEEEkeywords}

\section{Introduction}
\label{sec:introduction}
Designing green wireless networks \cite{lister-09,palicot-2005,greentouch-10} has become increasingly important for modern wireless networks, in particular, to manage operating costs. A challenge for modern (beyond 4G and 5G) cellular networks is not only to respond to the explosion of data rates, but also to manage network energy consumption. The concept of small cell networks appears as a good candidate solution to raise such a challenge (see e.g., \cite{hoydis-vtm-2011}). As small cell networks will be distributed to large extent and subject to high inter-cell interference, designing distributed energy-efficient interference management schemes appear as a natural need.

For being able to design green networks, an energy-efficiency (EE) metric is needed.
In \cite{goodman-pc-2000}, the EE of a communication between a transmitter and a receiver is defined as the ratio of the net data rate to the radiated power; the corresponding quantity is a measure of the average number of bits successfully received per joule consumed at the transmitter. Quite recently, there has been a resurgence of interest in this performance metric. There are several reasons for this and we will only provide a few of them. First, the EE as defined in \cite{goodman-pc-2000}, mathematically translates in a simple manner the trade-off between the benefit of increasing the transmit power in terms of data rate, and the induced cost in terms of consumed energy or amount of created interference. Second, as motivated in \cite{tse-2010}, there are applications in which the allowable delay is not tightly constrained. Therefore, the data rate is a less relevant measure than the energy needed to transmit the information and EE naturally appears as a metric to be optimized. We furthermore explain in this paper (Sec. \ref{sec:const}) why maximizing EE amounts to minimizing the total energy consumed by the transmitter when packet retransmission is considered.

Remarkably, the energy-efficiency metric proposed in \cite{goodman-pc-2000} possesses a good mathematical structure for optimization, especially for the distributed case, which partly explains why it has been applied in a large variety of scenarios of practical interest. Some examples are as follows. In \cite{meshkatiCDMA06}, it is applied to design a power allocation scheme in distributed multi-carrier CDMA (code division multiple access systems) systems by using a static non-cooperative game model (just as \cite{goodman-pc-2000}). In \cite{Mesh09}, it is used to account for the users delay requirements in energy-efficient wireless systems. In \cite{buzzi-2012} and \cite{valuetools}, also based on a static game model, the authors used the metric under consideration for sub-carrier assignment in distributed OFDMA (orthogonal frequency division multiple access) multicellular networks. In \cite{bacci-tsp-2013}, the authors study the problem of energy-efficient contention-based synchronization in OFDMA systems. In \cite{veronica-2009,belmega-tsp-2010}, the authors study the problem of pre-coding in MIMO (multiple input multiple output) point-to-point communication systems. In \cite{letreust-tvt-2013}, the EE metric is exploited to study the impact of sensing in terms of EE in cognitive radio networks.


Although fully relevant, the performance metric introduced in \cite{goodman-pc-2000} and used in the related works (this, in particular, includes those cited above) has left several issues unexplored, which has motivated the work reported here. In \cite{goodman-pc-2000}, and all related references known to the authors, the numerator of EE is (up to a constant) a packet success rate which only accounts for packet losses due to bad channel conditions. In the present work, we propose a significant generalization of the metric used in the aforementioned works to the case where packets are buffered in a finite size queue. Therefore, the packet loss due to overflows is also taken into account. On the other hand, we will show in Sec. III that accounting for this effect is relevant in terms of EE, only when the transmitter has a cost in terms of consumed power independent of the radiated power; this means that the transmitter consumes power even while waiting for new packets to arrive. It turns out that this is precisely what happens for most wireless transmitters. Indeed, the transmitter energy consumption is not only induced by the radiated power but also results from other causes such as the transmitter supply consumption \cite{richter}. Note that the authors of \cite{betz} were the first to consider a transmission cost of the type ``radiated power + constant'' to design distributed power control strategies for multiple access channels; in their model, the constant represents the computation power at the receiver. Our approach is markedly different from \cite{betz}, not only because the problem is tackled from a cross-layer perspective, but we also consider the more general case of distributed interference networks with a quality of service (QoS) constraint. For this purpose, two different models for the packet arrival rate are considered: 1) The quite simple model where the arrival rate is a constant (which is referred to as CAR for constant arrival rate). This case is useful e.g., for real-time traffic like video or streaming; 2) The more interesting model in which the arrival rate is related to the SINR (signal-to-noise plus interference ratio) through a quite generic relationship, is more suited to delay tolerant traffic like file transfer and adaptive rate services like WebRTC \cite{Google}; this case is referred to as AAR for adaptive arrival rate.

The main contributions of this paper can be summarized as follows\footnote{Note that preliminary results were presented in \cite{varma-icc-2012}.}:

\begin{enumerate}
\item To the best of the author's knowledge, this is the first time that the EE performance metric originally introduced in \cite{goodman-pc-2000} is generalized to a cross-layer approach, taking into account, the effects of the presence of a queue with finite size at the transmitter.
\item Apart from a few exceptions (which includes \cite{betz,varma-icc-2012,zappone-twc-2013}), all related works using EE in the sense of \cite{goodman-pc-2000} only consider the radiated power while, here, the total power consumed by the transmitter is taken into account. Since an affine relation between the radiated power and the total power is assumed \cite{richter}, this might seem as an incremental change but the presence of this fixed cost is the key ingredient which makes the cross-layer analysis fully relevant;
\item We apply the Nash equilibrium (NE) analysis (methodology reviewed in \cite{lasaulce-book} for example) in the presence of the two aforementioned features. Although the equilibrium analysis methodology is not new in itself, conducting it requires several non-trivial results to be proved, including the proof of quasi-concavity of the derived performance metric. Indeed, quite surprisingly, both in the case of CAR and AAR, it can be shown to be quasi-concave with respect to the radiated power. Similarly, even though the more general performance metric under investigation is seemingly quite complex, the important property of standardness for the best-responses can be shown to be retained. This guarantees both NE uniqueness and the convergence of relevant distributed optimization algorithms (such as the one used here) to this equilibrium;
\item Apart from a few exceptions (which includes \cite{Mesh09}), all related works using EE in the sense of \cite{goodman-pc-2000} do not consider the QoS aspect. The QoS constraint is precisely considered here in the case of the CAR protocol (note that the AAR protocol is assumed to automatically controls the packet rates by observing the packet loss). The corresponding game has semi-continuous payoffs and this seems to be the first instance of an energy-efficient power control game to be identified as being semi-continuous, which allows us to prove the existence of an equilibrium by exploiting a fixed point theorem from \cite{dasgupta-1986};
\item A thorough numerical analysis is provided to assess the benefits from taking the presence of a queue with finite size into account and to give new insights into designing energy-efficient communications systems. These insights include assessing gains in terms of energy consumed by the whole transmit device, following the interpretation we provide of energy-efficiency maximization as energy minimization in Sec. \ref{sec:const}.
\end{enumerate}

This paper is structured as follows. In Sec. II, we present the general system model. In Sec. III, we construct the proposed performance metric highlighting contributions 1) and 2). In Sec. IV, we define the two power control games of interest and conduct the equilibrium analysis, which is essential to characterize the convergence of the suggested distributed algorithm (existence and uniqueness of the convergence points), highlighting contributions 3) and 4). Sec. V highlights interesting numerical results that support the proposed approach, i.e, 5). Finally, we conclude the paper and several extensions of this work are provided.

\section{System model}
\label{sec:system model}

The purpose of this section is to describe the communication model considered for cross-layer energy-efficient power control, which consists in expressing the SINR and packet arrival rate for a given user. A general interference network is considered with $N$ transmitter-receiver pairs, in which each transmitter communicates with its respective receiver, while under interference from the other transmitters \cite{carleial-1978}. Let $\mc{N}=\{1,2,...,N\}$ be the set of transmitters. Transmitter $i \in \mc{N}$ transmits with power level $p_i \in [0, P_{\max}]$, where $P_{\max} > 0$ is the maximum possible transmit power, which is identical for all transmitters (the analysis does not lose its generality with this assumption). The vector $\ul{p} = \left(p_1,p_2,\dots,p_N \right)$ will be referred to as the power or action profile on the current data block or packet. We also denote by $\ul{p}_{-i}$, the $(N-1)$ dimensional vector obtained by removing the $i^{\mathrm{th}}$ component from $\ul{p}$. For notational simplicity, we also sometimes represent $\ul{p}$ as $(p_i,\ul{p}_{-i})$, when the dependence of certain functions on $p_i$ has to be shown explicitly. By transmitting at $p_i$, each user $i$ has a resulting SINR $\gamma_i$ at his receiver of interest which is a function of $\ul{p}$, and is assumed to be given by:
\begin{equation}
\gamma_i(\ul{p}) = \frac{p_i g_{ii}}{\sigma_i^2 + \ds{\sum_{j=1, j\neq i}^N p_j g_{ji}}}
\label{eq:gamma}
\end{equation}
where $g_{ji}$ represents the quasi-static or block fading channel gain of the link between transmitter $j$ and receiver
$i$ on a given band, $\sigma_i^2 = \sigma^2$ is the variance of the Gaussian noise at receiver $i$ (these variances can be assumed to be equal without any loss of mathematical generality). In wireless systems such as those being implemented in recent cellular system standards, packets arrive from an upper layer (e.g. IP layer) following an arrival rate that is related to the SINR. In this paper, we assume that the packet arrival process follows a Bernoulli process with probability $q_{\mathrm{X}}(\gamma_i(\ul{p}))$ where $X \in \left\{\mathrm{CAR}, \mathrm{AAR} \right\}$; this corresponds to the classical ON/OFF sources \cite{takagi-85}. In the case of CAR, it trivially expresses as:
\begin{equation}
\forall i \in \mc{N}, \ q_{\mathrm{CAR}}(\gamma_i(\ul{p})) = q
\end{equation}
with $q \in [0,1]$. This is best used for real-time applications where delay is not tolerable, however, in some applications this packet arrival model is not suitable. For instance, this is the case for applications such as file transfer or browsing. In such a situation, there is no constant stream of data and so the arrival rate can be optimized for best performance in terms of data rate and QoS. This is one of the reasons why we also investigate the case of AAR for which we assume that the arrival rate is given by:
 \begin{eqnarray}
\forall i \in \mc{N}, \ q_{\mathrm{AAR}}(\gamma_i(\ul{p}))=g(\Phi_{\mathrm{AAR}}(\gamma_i(\ul{p})) )
\label{eq:tcp}
\end{eqnarray}
where $\Phi_{\mathrm{AAR}}$ is the packet loss function and $g$ is a function which is assumed to be continuous, invertible, and has an inverse function $g^{-1}$ which is twice differentiable, decreasing, and convex. The goal, by making these assumptions, is to make ``compatible'' the implemented AAR scheme with the fact that power control is distributed. Indeed, the existence of an NE is highly desirable for power control since this means the distributed algorithm will stabilize to a point which is predictable and whose performance can be assessed. A sufficient condition for having the existence of a pure NE is that every payoff function be quasi-concave w.r.t. to the individual strategy. The properties assumed for $g$ and its inverse corresponds to a class of AAR schemes which guarantees quasi-concavity. It can be verified that some important models on packet rate control satisfy these properties. To provide a specific example, the widely used and very useful approximation of the arrival rate process for the Transmission Control Protocol (TCP), which is due to \cite{Padhye98}, verifies these conditions. Therein, $g$ is merely given by $g(\Phi) = \frac{\kappa}{\sqrt{\Phi}}$, where $\kappa \in [0,1] $ is a parameter which depends on the system design and the round trip time. The resulting rate can be interpreted as the average value for the rate. It is very likely that our assumptions can be refined but they can be  seen as a first attempt towards characterizing the classes of AAR schemes which can operate in a harmonious manner with distributed power control schemes.

\emph{Remark 1.} The CAR protocol can also be seen as a constant piece-wise approximation of any adaptive arrival rate protocol in which arrival rate variations are much more slower than channel variations. On the other hand, the AAR case aims at better understanding more complex scenarios where both arrival rate and channel variations have quite similar time-scales. This is close to WebRTC congestion control protocols, like the one proposed by Google \cite{Google}, where the sending rate is adapted based on the observed packet loss \cite{Singh}.

\emph{Remark 2.} It would be possible to study a more general communication scenario by considering multi-band communications, MIMO communications, a more advanced reception scheme (e.g., interference cancellation as in \cite{lasaulce-twc-2009}), and by integrating more the medium access control protocol (such as a carrier sense multiple access -CSMA- protocol). There are several reasons why we do not treat these scenarios here. First, we want to emphasize in a manner as clear as possible the real contributions of this paper namely, the introduction of a queue for the problem of energy-efficient power control. The objective is to enable all the transmitters in the system to operate in the same spectrum or carrier at the same time by managing the interference level; This is very useful e.g., for cellular systems with intercell interference and emerging applications such as small cells networks. Second, studying the power control problem is the main step towards these extensions. For instance, in \cite{meshkatiCDMA06} in which the authors address EE over multi-carrier multiple access channels, it is proved that the best selfish/equilibrium policy for a transmitter is to select its best carrier (in terms of SINR) and apply the single-carrier policy to tune the power level over this carrier. Therefore, the assumed model can be understood as a single-band model (e.g., several base stations which try to mitigate inter-cell interference on a given band) or a multi-band model for which interference is managed for the selected channel or interference management is performed independently from band to band.

\emph{Remark 3.} Note that we do not assume the presence of a ``central'' node which would implement multiuser scheduling. The instantaneous (individual) channel gain need not assumed to be known to a given transmitter and assuming a user selection scheme would require more coordination and more knowledge in terms of channel state information (e.g., global channel state information) whereas this is what we try to avoid for the considered framework namely, distributed power control for interference networks.

\section{A new energy-efficiency performance metric}
\label{sec:performance-metric}

\subsection{Construction}
\label{sec:const}
In \cite{goodman-pc-2000}, EE is defined as the ratio between the average net data transmission rate and the power consumed for sending a given packet. When the radiated power is considered as the transmission cost, this ratio merely equals $\frac{R f(\gamma_i(\ul{p}) )}{p_i}$. The quantity $R$ is the gross data rate on the radio interface. In this paper, it can be seen as the ``target rate'' which is typically determined by the bandwidth and modulation-coding scheme (MCS) and would be achieved if there were no transmission error (i.e., the block error rate $1-f(.)$ would be vanishing). Therefore, the frequency at which the power control is performed is assumed to be higher than the frequency at which the rate is adapted, which is an important scenario in practice. More justifications can be found e.g., in \cite{goodman-pc-2000,Mesh09}. Other interpretations of $R$ are available in the literature, see e.g., \cite{belmega-tsp-2010}\cite{zappone-twc-2013}. Each packet transmitted on the channel is received without any errors with a probability which depends on the quality of the communication link, the interference, and transmit power levels. The corresponding block or packet success rate (also called efficiency function) is precisely the function $f(\gamma_i(\ul{p}))$ above. The function $f:[0,+\infty) \rightarrow [0,1]$ is a sigmoidal\footnote{A sigmoidal function is a function which is initially convex for $\gamma \in [0,\gamma_+]$ and eventually concave for $\gamma \in [\gamma_+,\infty)$.} or S-shaped function verifying $f(0)=0$ and $\ds{\lim_{x \to \infty}} f(x)=1$ (see \cite{rodriguez} for more details). Common examples for $f$ are $f(x) = (1-e^{-x})^M$ \cite{Mesh09}, $f(x) = e^{-\frac{c}{x}}$ \cite{veronica-2009,belmega-tsp-2010}, where $M\geq1$ is the packet length and $c>0$ is some constant related to spectral efficiency (this relation is specified in Sec. \ref{sec:numerical-results}). Energy-efficiency is particularly relevant when packet re-transmission is allowed. When there is no re-transmission, the energy\footnote{Here, the energy under consideration is the energy associated with the radiated signal.} consumed to send $V$ bits while transmitting at the power level $p_i$ is $p_i \frac{V}{R}$. Minimizing energy amounts to minimizing $p_i$. However, when re-transmission is allowed (typically by using an automatic repeat request -ARQ- protocol, that is used at the physical layer independently of the architecture at the upper layer), the average duration to send a packet equals $\frac{V }{R f(\gamma_i(\ul{p}))}$ and the energy consumed becomes $p_i \frac{V }{R f(\gamma_i(\ul{p}))}$. Clearly, minimizing energy amounts to maximizing EE. This means that, at least in presence of re-transmissions, the classical approach which consists in minimizing $p_i$ (subject to some QoS constraints) induces a loss in terms of minimizing the energy consumption; this will be illustrated in Sec. \ref{sec:numerical-results}. In the scenario investigated in this paper, the fact that both the total power consumed by the transmitter and the presence of a packet buffer with finite size are considered makes the construction of energy-efficiency more involving than the aforementioned derivation.

A simple model which allows one to relate the radiated power to the total consumed power is provided in \cite{richter}; it is given by $P_{\mathrm{total},i} = a p_i + b$, where $a \geq 1, b \geq 0$ are some parameters. We will assume without loss of generality that $a=1$. The quantity $b$ precisely represents the consumed power when the radiated power is zero\footnote{This power consumption occurs even when data is not transmitted due to various causes such as pilot signaling, power amplifier consumption, cooling costs, etc.}. Let $\Phi_{\mathrm{X}}$, $X \in \{\mathrm{CAR}, \mathrm{AAR} \}$, represents the packet loss due to both bad channel conditions and packet buffer finiteness (more details about this is provided a little further). For each packet that enters the queue, $\frac{1}{f(\gamma_i)}$ attempts \footnote{For the sake of clarity, here and in other places in the paper, $\ul{p}$ is omitted from the notations.} on transmitting it is made, as the transmission is successful only with a probability $f(\gamma_i)$. Additionally, from the previous equations, we have shown that, on average, $q_{\mathrm{X}}(\gamma_i)[1-\Phi_{\mathrm{X}}(\gamma_i)]$ packets come out of the queue. Hence, the average number of transmission attempts is given by $\frac{q_{\mathrm{X}}(\gamma_i)[1-\phi_{\mathrm{X}}(\gamma_i)]}{f(\gamma_i)}$ leading to an expected power cost of $b + p_i \frac{q_{\mathrm{X}}(\gamma_i)   [1-\Phi_{\mathrm{X}}(\gamma_i)]}{f( \gamma_i)}$. Since the net data rate or goodput is given by $R q_{\mathrm{X}}(\gamma_i)  [1-\Phi_{\mathrm{X}}(\gamma_i)]$, we are now able to define the EE metric $\eta_{i,\mathrm{X}}(\ul{p})$ as the ratio between the average net data transmission rate and the average power consumption, which gives:
\begin{equation}
\label{eq:eta}
\eta_{i,X}(\ul{p}) =  R \frac{   q_{\mathrm{X}}(\gamma_i(\ul{p})    )\left[1
-\Phi_{\mathrm{X}}(\gamma_i(\ul{p}) )\right]  }{   b + p_i \frac{q_{\mathrm{X}}(\gamma_i(\ul{p}))   \left[1-\Phi_{\mathrm{X}}(\gamma_i(\ul{p})) \right]}{f( \gamma_i)}}.
\end{equation}
This definition shows that the cross-layer design approach of power control is fully relevant in terms of EE when the transmitter has a cost, which is independent of the radiated power; otherwise, when $b=0$, one falls into the original framework of \cite{goodman-pc-2000}. On the other hand, when $b$ is large, the EE function behaves like a packet success rate function.

Although the efficiency function $f$ (which is assumed to be sigmoidal) can be easily related to the SINR through simple functions such as those mentioned previously, expressing the packet loss function is more involving. Relating $\Phi_{\mathrm{X}}$ to the SINR is the purpose of what follows. A packet is declared to be lost (blocked) only if a new packet arrives when the packet buffer is full and, on the same time-slot, transmission of the packet on the radio interface failed. Note that these two events are independent because the event of ``transmit or not'' for the current packet on the radio interface, does not impact the current size of the queue, but only the one for the next time slot. This amounts to considering that a packet coming at time slot $t$, is rejected at the end of time slot $t$, the packet of the radio interface having not been successfully transmitted. By considering the stationary regime of the queue and assuming the protocol $\mathrm{X}$, the fraction of lost packets $\Phi_{\mathrm{X}}$ can be expressed as follows:
\begin{equation}
\Phi_{\mathrm{X}}(\gamma_i(\ul{p})) =  [ 1 - f(\gamma_i(\ul{p})) ] \Pi_{\mathrm{X}}(\gamma_i(\ul{p}))
\label{eq:phi}
\end{equation}
where $\Pi_{\mathrm{X}}(\gamma_i)$ is the stationary probability that the packet buffer is full. Indeed, as already mentioned, each transmitter is assumed to be equipped with a device that allows the packets to be stored in a memory buffer (of size $K \geq 1$) before transmission. Packets arrive into the buffer and get transmitted through a queuing process at the buffer. Denote by $Q_{i,t}$ the size of the queue for transmitter $i$ at time slot $t$. The size of the queue $Q_{i,t}$ is a Markov process on the state space $\mc{Q}_i=\{0,1,\ldots,K\}$. It is known (see \cite{wolff} for example) that in the stationary regime of the stochastic process $Q_{i,t}$ the probability that the size of the queue equals $K$ is given by:
\begin{equation}
\Pi_{\mathrm{X}}(\gamma_i(\ul{p})) = \frac{\omega_{\mathrm{X}}^K(\gamma_i(\ul{p}))}{1+\omega_{\mathrm{X}}(\gamma_i(\ul{p}))  +\ldots+\omega_{\mathrm{X}}^K(\gamma_i(\ul{p}))}
\label{eq:pi}
\end{equation}
with
\begin{equation} \label{eq:omega}
\omega_{\mathrm{X}}(\gamma_i(\ul{p})) =  \frac{q_{\mathrm{X}}(\gamma_i(\ul{p}))    \left[1
-f(\gamma_i(\ul{p}))\right]}{\left[ 1-q_{\mathrm{X}}(\gamma_i(\ul{p})) \right] f(\gamma_i(\ul{p}))}
\end{equation}
where $X \in \left\{\mathrm{CAR}, \mathrm{AAR} \right\}$.

In the case of $X=\mathrm{AAR}$, the packet arrival rate $q_{\mathrm{AAR}}$ is a function of the packet loss and the packet loss, a function of $q_{\mathrm{AAR}}$. The following proposition ensures that the AAR process achieves an average packet arrival rate according to the following proposition. For the purpose of making the inter-dependency of the two following equations clear, we express explicitly in these equations, some of the parameters used implicitly in the rest of the paper.
\begin{proposition}
The packet arrival rate $q_{\mathrm{AAR}}$ is obtained as the unique fixed point of these equations:
\begin{equation}
\Phi_{\mathrm{AAR}}(\gamma_i(\ul{p}))= (1-f(\gamma_i(\ul{p}))) \Pi_{\mathrm{AAR}} (\gamma_i(\ul{p}))
\end{equation}
where $\Pi_{\mathrm{AAR}} (\gamma_i(\ul{p}))$ has $q_{\mathrm{AAR}}$ as a parameter as seen from (\ref{eq:omega}) and (\ref{eq:pi}), and:
\begin{equation}
q_{\mathrm{AAR}}(\Phi_{\mathrm{AAR}})=g(\Phi_{\mathrm{AAR}}).
\end{equation}
\end{proposition}

\begin{proof} It can be verified that the two equations are continuous and differentiable. The packet arrival rate $q_{\mathrm{AAR}}(\gamma_i)$ ranges from $0 <g(0) \leq 1$ to $0\leq g(1)<g(0)$ and $\Phi_{\mathrm{AAR}} (\gamma_i)$ ranges from $0$ to $1$. Based on the properties of $\Phi_{\mathrm{AAR}}$ given in App. \ref{app:qc}, $\Phi_{\mathrm{AAR}}$ ranges from $0$ to $1$ as $q$ goes from $0$ to $1$. Now study $F(q_{\mathrm{AAR}}) \triangleq \Phi_{\mathrm{AAR}}(\gamma_i,q_{\mathrm{AAR}}) - g^{-1}(q_{\mathrm{AAR}})$. The function $F(q_{\mathrm{AAR}})$ is a continuous and differentiable function in the interval of $q \in [0,1]$. A point such that $F(q_{\mathrm{AAR}})=0$ is a fixed point for this set of equations. Based on the mean value theorem \cite{meanvalue}, and from the limits $\lim_{q \to 0} F(q_{\mathrm{AAR}}) \leq 0$ and $\lim_{q \to 1} F(q_{\mathrm{AAR}}) \geq 0$, we have $F(q_{\mathrm{AAR}})=0$ for some $q \in [0,1]$. Also note that $F(q)$ is strictly increasing and so the point where $F(q_{\mathrm{AAR}})=0$ is unique.

The fixed point equation can be solved as:
\begin{equation}
\label{eq:tcpfpe}
g^{-1}(q_{\mathrm{AAR}}(\gamma_i(\ul{p}))) = [1-f(\gamma_i(\ul{p}))] \frac{ \omega_i(\gamma_i(\ul{p}))^K} {\sum_{j=0}^K \omega_i(\gamma_i(\ul{p}))^j }
\end{equation}
and has a unique solution.
\end{proof}

\emph{Remark 3.} For $b>0$, it can be seen that for any $X\in \{\mathrm{CAR}, \mathrm{AAR}\}$, $q_{\mathrm{X}} \rightarrow 1 \Rightarrow \omega_{\mathrm{X}} \rightarrow +\infty \Rightarrow
\Pi_{\mathrm{X}} \rightarrow 1 \Rightarrow \Phi_{\mathrm{X}} \rightarrow 1-f$, which means that one falls into the framework of \cite{betz}.

\emph{Remark 4.} When the packet arrival is constant (i.e., $\mathrm{X} = \mathrm{CAR}$), the dependency of $\Pi_{\mathrm{X}}$ regarding the SINR follows a simple relation. However, when the AAR protocol is assumed, the relationship is less trivial. Indeed, the packet loss $\Phi_{\mathrm{X}}$ depends on $\omega_{\mathrm{X}}$ through (\ref{eq:phi}) and (\ref{eq:pi}). The quantity $\omega_{\mathrm{X}}$ depends on the arrival rate  $q_{\mathrm{X}}$. But, in the AAR case, $q_{\mathrm{X}}$ also depends on the packet loss. This is the reason why we assume that, under the AAR protocol assumption, each transmitter operates at the fixed point associated with the aforementioned dependency chain. Therefore, this amounts to fixing the packet loss function to have a certain form. AAR can thus be seen as an indirect way of imposing a certain QoS on the transmission. To be more specific, if one assumes an arrival rate process which can be approximated as in \cite{Padhye98} (namely, $g(\Phi) = \frac{\kappa}{\sqrt{\Phi}}$) and the regime of large buffer size $K\rightarrow \infty$, the operating packet arrival rate function can be shown to be~:
\begin{equation}
 \lim_{K \to \infty} q_{\mathrm{AAR}}(\gamma_i) = f(\gamma_i) \frac{1+ \sqrt{1+ 4 \left( \frac{\kappa}{f(\gamma_i) } \right)^2}}{2}.
\end{equation}

\emph{Remark 5.} In the above equations, we have implicitly made a symmetry
 assumption: the efficiency and arrival rate functions are assumed to be identical for all users. This choice allows one to gain in terms of clarity while the extension to the non-symmetric case is ready. For the same reason, the gross data rates at which the users transmit $R_i, i\in \mc{N}$, have been assumed to be equal (to $R$ bit/s).

\emph{Remark 6.} As the form of the performance metric under consideration implicitly indicates (see (\ref{eq:eta})), the choice made in this paper is not to account for possible memory effects which would be due e.g., to correlated channel realizations from block to block or the state of the queue. This choice is coherent with the related literature on EE which originates from \cite{goodman-pc-2000} and the merit of it is that the corresponding power control policies remain distributed in the sense of the required knowledge to implement it. As seen in Sec. \ref{sec:equilibrium-analysis}, a transmitter only needs to know its instantaneous SINR to tune its power level on the current block and therefore manage EE and created interference. Exploiting stochastic models can be seen as a relevant extension of the present paper which would lead to a better performance (provided all the additional parameters required are well estimated) but at the expense of obtaining power control policies which are (possibly much) more demanding computationally and requiring (possibly much) more information (see e.g., \cite{letreust-isccsp-2010}\cite{destounis}). Summarizing, the proposed approach can be seen as a reasonable tradeoff between performance gain in terms of EE and ease of implementation.

\subsection{Properties}\label{sec:properties}

In order to obtain more insights about the impact of having a buffer on energy-efficiency, we now briefly analyze the case of CAR. The following result holds.

\begin{proposition} Let $\mathrm{X}=\mathrm{CAR}$. For all $i\in \mc{N}$, the EE function $\eta_{i,\mathrm{X}}$ is a strictly increasing function of the parameter $q$.
\end{proposition}

\begin{proof}
Let $\ul{p}$ be fixed and remove the dependency toward $\ul{p}$ and $\gamma_i$ from the notations. The EE function can be rewritten as $\eta_{i,\mathrm{CAR}} = \frac{1}{\frac{b}{(1-\Phi_{\mathrm{CAR}})q} + \frac{p_i}{f}}$. Clearly, if the sufficient condition $\frac{\partial\Phi_{\mathrm{CAR}}}{\partial q} < \frac{1-\Phi_{\mathrm{CAR}} }{q}$ holds, then $\frac{\partial}{\partial q} (1-\Phi_{\mathrm{CAR}})q > 0 $. From this, it follows that $\frac{\partial \eta_{i,\mathrm{CAR}}}{\partial q} > 0$. Let us prove the sufficient condition. The derivative $\frac{\partial\Phi_{\mathrm{CAR}}  }{\partial q}$ can also be written as $\frac{\partial\Phi_{\mathrm{CAR}}}{\partial q}=-\Phi_{\mathrm{CAR}}^2\frac{\partial(\Phi_{\mathrm{CAR}}^{-1})}{\partial q}$ with $\Phi_{\mathrm{CAR}}^{-1}=1+\frac{1}{\omega_{\mathrm{CAR}}} +..+\frac{1}{\omega_{\mathrm{CAR}}^K}$. Using
\begin{equation}
\frac{\partial\omega_{\mathrm{CAR}}  }{\partial q} = \frac{1-f}{f} \frac{1}{(1-q)^2}
\end{equation}
implies that
\begin{equation}
\frac{\partial\Phi_{\mathrm{CAR}}^{-1}}{\partial q}=\left(\frac{1}{\omega_{\mathrm{CAR}}} +..+\frac{K}{\omega_{\mathrm{CAR}}^{K}} \right)\frac{1}{q(1-q)} > \frac{1}{\Pi_{\mathrm{CAR}} q (1-q)}.
\end{equation}
The sufficient condition follows by using $\frac{1}{\Pi_{\mathrm{CAR}} q (1-q)}>0$ and thus $\frac{\partial\Phi_{\mathrm{CAR}}  }{\partial q} <0< \frac{1-\Phi_{\mathrm{CAR}}}{q}$.
\end{proof}

This proposition mathematically translates the following intuition. If the packet arrival rate $q$ decreases, the average duration during which the buffer is empty increases. Since there is a fixed transmission cost $b$, this induces a waste of energy.

To conclude this section, let us analyze the limit of large buffer size. Two sub-cases can be distinguished.
\begin{itemize}
\item Case 1: $ q_{\mathrm{X}} > f(\gamma_i(\ul{p}))$, i.e., $\omega_{\mathrm{X}}  >1$. We have that the steady-state probability of having a full buffer $\ds{\lim_{K \to \infty}} \Pi_{\mathrm{X}} = \frac{\omega_{\mathrm{X}}  -1}{\omega_{\mathrm{X}}} $ and a simplification yields $\Phi_{\mathrm{X}} = 1-\frac{f(\gamma_i(\ul{p}))}{q_{\mathrm{X}}}$. Thus the EE becomes $\ds{\lim_{K \to \infty}} \eta_{i,\mathrm{X}}(\ul{p})= \frac{R f(\gamma_i(\ul{p}))}{b + p_i}$. This means that a higher probability of entrance than exit causes the queue size to blow up, and there are always packets to be transmitted, which explains why one falls into the framework of \cite{betz} in Case 1.

\item Case 2: $q_{\mathrm{X}} \leq f(\gamma_i(\ul{p}))$, i.e., $\omega_{\mathrm{X}}  \leq 1$. If $f(\gamma_i(\ul{p}))=q_{\mathrm{X}}$, then $\Pi_{\mathrm{X}} = \frac{1}{K} $ and $\ds{\lim_{K \to \infty}} \Pi_{\mathrm{X}} =0$. For $f(\gamma_i(\ul{p}))>q_{\mathrm{X}}$, we have also that $\ds{\lim_{K \to \infty}} \Pi_{\mathrm{X}} =0 $ and then simplification yields $\Phi_{\mathrm{X}} \to 0$. Thus the EE becomes $\ds{\lim_{K \to \infty}} \eta_{i,\mathrm{X}}(\ul{p})= \frac{R}{\frac{b}{q_{\mathrm{X}}} + \frac{p_i}{f(\gamma_i(\ul{p}))}}$. This means that, even with the fixed consumption cost $b$, the EE performance metric to be optimized becomes $\frac{f(\gamma_i(\ul{p}))}{p_i}$ (i.e., the same metric as \cite{goodman-pc-2000}). This is also quite intuitive, as in the steady state, due to a higher probability of exit, the buffer is never full and there is no packet loss due to buffer overflow.
\end{itemize}

\emph{Remark 7.} The above special case analysis suggests that, in the regime of large buffer size, the power control policies may be obtained from an approximated payoff function which is simpler than the exact expression (\ref{eq:eta}). It is seen that, depending on the current value of the SINR and arrival rate, the approximated payoff function coincides either with that of \cite{goodman-pc-2000} or \cite{betz}.

\subsection{QoS constraint}

To conclude this part, we will mention how the QoS constraint is treated in our analysis. As already mentioned in Sec. I, one of the recurrent problems with most works using the performance metric introduced in \cite{goodman-pc-2000} is that EE can be maximized at a power level which does not guarantee a minimum QoS. This is why, in the case of CAR, we also consider a constraint when maximizing (\ref{eq:eta}): the packet loss rate $\Pi_{\mathrm{CAR}} [1-f(\gamma_i)]$ has to be less than an upper bound $\epsilon$. For example, in cellular systems, typical values for $\epsilon$ are $0.1$ or $0.01$, based on the system requirements. Adding this constraint restricts the range of power usable by the transmitter by adding a lower bound on the power. This lower bound depends on the entry probability $q$ and on the size of the queue $K$. At last, our choice is not to impose this constraint for the AAR protocol since, by construction, this protocol aims at automatically adapting the packet arrival rate to congestion.

\section{Equilibrium analysis and distributed power control algorithm}
 \label{sec:equilibrium-analysis}

Since it is assumed that transmitter $i$, $i\in \mc{N}$, can only control the variable $p_i$ of the $N-$variable function $\eta_{i,\mathrm{X}}(\ul{p})$ and is assumed to consider the energy-efficiency of his own communication, the power control problem is naturally distributed in terms of the decision. The ultimate goal of this section is to propose a power control algorithm which is distributed both in terms of the decision and information (only individual SINR feedback is required to adapt the power level). While the algorithm itself is directly inspired from existing works, its convergence analysis does not follow from a direct adaptation of existing results. We first begin by an equilibrium analysis of two non-cooperative games associated with the two considered power control scenarios (CAR and AAR), before proposing the algorithm and proving its convergence.

\subsection{Equilibrium analysis of the associated games}

A non-cooperative game under strategic form is merely given by an ordered triplet (see e.g., \cite{lasaulce-book}). With the notations of this paper it writes as
\begin{equation}
\mc{G}_{\mathrm{X}} = \left(\mc{N}, \left\{ \mc{P}_i \right\}_{i \in \mc{N}}, \left\{ u_{i,\mathrm{X}} \right\}_{i \in \mc{N}}      \right)
\end{equation}
where the set of decision-makers (DMs) or players is therefore the set of transmitters, the action space for DM $i$ is $\mc{P}_i = [0, P_{\max}]$, and $u_{i,\mathrm{X}}$ is the payoff function of DM $i$ when the arrival rate model is $X$. As explained in Sec. II, when CAR is assumed, a QoS constraint is imposed on the packet loss. Under this assumption, the payoff function is chosen to be:

\begin{equation}
u_{i,\mathrm{CAR}}(\ul{p}) = \left|
\begin{array}{cc}
\eta_{i, \mathrm{CAR}}(\ul{p}) & \ \mathrm{if}  \ \  \Phi_{\mathrm{CAR}}(\gamma_i(\ul{p})) \leq \epsilon\\
 \theta_i(\ul{p})       &    \mathrm{otherwise}
\label{eq:utiludp}
\end{array}.
\right.
\end{equation}
Here, $\theta_i(\ul{p})$ is a non-decreasing function of $p_i$ which satisfies: $\forall \ul{p}\in [0,P_{\max}]^K, \theta_i(\ul{p}) \leq \eta_{i, \mathrm{CAR}}(\ul{p})$. For example, $\theta_i(\ul{p}) = \ds{ R \frac{q\left[1-\Phi_{\mathrm{CAR}}(\gamma_i(\ul{p})) \right]}{b +P_{\max}}}$ is an appropriate choice. The latter choice has the important advantage of preserving the quasi-concavity property for the payoff (as proved in Appendix A). The corresponding payoff also represents an energy-efficiency measure and is measurable in many real systems (e.g., by counting the number of ACK/NACK). Elaborating further on the practical motivations, this choice also corresponds to an envisioned scenario for future 5G cellular systems. Therein, the first concern will likely be to guarantee a certain QoS to the user. Energy will probably appear as a secondary concern. Therefore, as long as QoS can be ensured, the communication terminal optimizes energy consumption or energy-efficiency. If the QoS target cannot be reached, then the goal is to make the service as robust (minimize outage or maximize the goodput). Naturally, the QoS constraint may not be satisfied even if goodput maximization is pursued. Usually, the corresponding feasibility analysis has to be conducted offline e.g., by the operator who deploys his network. Since meeting the constraint amounts to reaching a certain level in terms of SINR (note that $\Phi_{\mathrm{CAR}}$ is invertible), the feasibility analysis falls into a literature which is broader than the one on energy-efficient power control and corresponds to a problem which is not specific to the work presented here (for more details see e.g., \cite{Bacelli}). Some numerical results are provided in Sec. \ref{sec:numerical-results} to provide some information about feasibility.

For the AAR case, the payoff function is simply defined by
\begin{equation}
u_{i,\mathrm{AAR}}(\ul{p}) = \eta_{i, \mathrm{AAR}}(\ul{p}).
\end{equation}

A fundamental solution concept for a non-cooperative game is the Nash equilibrium. There are at least two important reasons for this. When operating at an NE, a network possesses a form of strategic stability: a transmitter which changes his power control policy while the others keep on using the equilibrium policies, will see this payoff decreased or maintained in the best case. The second reason is that, under some conditions, important iterative distributed optimization algorithms such as the sequential best-response dynamics (called sequential iterative water-filling in the literature of distributed power allocation whose objective is to maximize the transmission rate) converge to an NE. It turns out that the two games under study verify a simple sufficient condition which allows the second feature to be exploited. All of this gives us a strong reason for conducting the equilibrium analysis for the two defined games in order to show that the two games above possess an equilibrium, to prove that it is unique and to provide a simple distributed optimization algorithm which converges to it. This is the purpose of the following propositions which follow the definition of a Nash equilibrium in our context.

\begin{definition} The vector of transmit power levels $\ul{p}_{\mathrm{X}}^{\mathrm{NE}}$ is a pure Nash equilibrium of the game $\mc{G}_{\mathrm{X}}$ if:
\begin{equation}
\forall i \in \mc{N}, \ \forall p_i \in \mc{P}_i, \ u_{i,\mathrm{X}}(\ul{p}^\mathrm{NE})
\geq u_{i,\mathrm{X}}(p_i, \ul{p}_{-i,\mathrm{X}}^{\mathrm{NE}}).
\end{equation}
\end{definition}

\begin{proposition} For $X \in \left\{ \mathrm{CAR}, \mathrm{AAR} \right\}$, the game $\mc{G}_{\mathrm{X}}$ admits at least one pure Nash equilibrium.
\end{proposition}

\begin{proof} The proof is based on a fixed point theorem proved in \cite{dasgupta-1986}. The called theorem states that if the action spaces are compact convex sets, every payoff function of the game is upper semi-continuous w.r.t. the action profile ($\ul{p}$ in our context), and for any DM $i$ the payoff function is quasi-concave w.r.t. to the individual action ($p_i$ in our context), then the game possesses a pure Nash equilibrium. For $X \in \left\{ \mathrm{CAR}, \mathrm{AAR} \right\}$, the action space is $[0,P_{\max}]$ which is a compact convex space. When $X = \mathrm{CAR}$, $u_{i, \mathrm{X}}$ is upper semi-continuous w.r.t $p_i$ whereas it is continuous when $X = \mathrm{AAR}$. For all $i\in \mc{N}$, the EE function $\eta_{i,\mathrm{X}}(\ul{p})$ is quasi-concave w.r.t. $p_i$ and has a unique maximum point denoted by $p^*_i(\ul{p}_{-i})$. The proof for this relies, in particular, on the sigmoidness assumption for $f$ and can be found in App. \ref{app:qc}. The specific case of CAR is analyzed in App. \ref{app:qccar} and App. \ref{app:gepsilon}, and the AAR case is analyzed in App. \ref{app:qctcp}.

\end{proof}

To our knowledge, all related works on energy-efficient power control use payoffs which are continuous with the power profile $\ul{p}$. Interestingly, a relevant power control game in which continuity is not available can be exhibited for the case of $X= \mathrm{CAR}$.

\begin{proposition} For $X \in \left\{ \mathrm{CAR}, \mathrm{AAR} \right\}$, the game $\mc{G}_{\mathrm{X}}$ admits a unique pure Nash equilibrium, for which the equilibrium power policy will be denoted by $\ul{p}_{\mathrm{X}}^{\mathrm{NE}}$.
\end{proposition}

\begin{proof}
The proof of this result mainly relies on one important property of the studied games, namely both games are standard in the sense of \cite{yates}. In App. \ref{app:standard} we prove that the DMs' best-responses are always standard functions; by definition, the best-response of a DM $i$ to the (reduced) action profile $\ul{p}_{-i}$ is the set-valued function defined by $\mathrm{BR}_{i,\mathrm{X}}(\ul{p}_{-i})=\arg \ds{ \max_{p_i} } u_i(\ul{p})$. If the best-responses of all the DM's are standard, then the game is also standard, which completes our proof.
\end{proof}

\subsection{The proposed implementation of the best-response algorithm}

In this section, the process of implementing the well-known sequential best-response dynamics is proposed for the game under investigation. The standardness property of the best-responses, which is exploited to prove the previous proposition is also sufficient to guarantee convergence of some important distributed optimization algorithms. Note that the argmax set mentioned in the proof is a singleton (a scalar value), which can be checked from App. \ref{app:gepsilon} for CAR and App. \ref{app:qctcp} for AAR. While the standardness property is available for the scenario studied in \cite{goodman-pc-2000} and many related works, it is seen here that, although the proposed QoS oriented cross-layer approach leads us to more complex and more general payoffs, this property turns out to be valid in the more general case under investigation.

This means that for these algorithms, in addition to convergence being ensured, the convergence point is also known to be unique. This is very useful to characterize the performance of a distributed power control algorithm. Here, we only mention one of such algorithms: the asynchronous or sequential best-response dynamics. This algorithm is well-known in game theory \cite{Fudenberg-1991} and draws its roots from the paper by Cournot \cite{cournot-1838}. It has been used in \cite{goodman-pc-2000} and is often used because convergence to the NE can be guaranteed under simple conditions such as standardness. Thus, we don't provide here a new type of algorithms but rather fully describe its implementation in the more general setup of this paper. Let $\ul{p}_{\mathrm{X}}^{\mathrm{NE}}$ be the unique NE of $\mc{G}_{\mathrm{X}}$. For the algorithm, we define $\ul{p}^t$ as the power control policy in the previous time-slot, and $\ul{p}^{t+1}$ as the power control policy for the current time slot. Algorithm \ref{alg:br} implements the sequential best-response dynamics for $\mc{G}_X$:

\begin{algorithm}                      
\caption{Sequential best-response dynamics}          
\label{alg:br}                           
\begin{algorithmic}                    
    \State $\Delta \leftarrow 2 \delta $ \Comment{Initialize the observed difference in power levels over time, $\delta$ is the tolerance.}
    \State $\ul{p}^0\leftarrow (P_{\max},P_{\max},\dots,P_{\max}) $ \Comment{The starting power levels can be chosen arbitrarily.}
\State $t \leftarrow 0 $ \Comment{The starting time is 0.}
     \While{$ \Delta  \geq \delta$} \Comment{The outer loop that iterates till the power policies converge.}
  \For{$i = 1 \to N$} \Comment{The inner loop iterating over the DM indices.}
\State $\Gamma_i = \ds{\frac{\gamma_i(\ul{p}^{t})}{p_i^{t}}}$ \Comment{Using the SINR feedback from its receiver, DM i calculates the interference term $\Gamma_i$ for the previous time slot.}
\State $p^* \leftarrow \arg \max_p \left( \ds{\frac{Rq_\mathrm{X}(p\Gamma_i ) (1-\Phi_\mathrm{X}(p\Gamma_i ))}{b + \ds{\frac{p}{f(p \Gamma_i)}} q_\mathrm{X}(p\Gamma_i ) (1-\Phi_\mathrm{X}(p \Gamma_i)) } } \right)$ \Comment{Calculate the optimal power that maximizes the EE.}

\If{X$=$CAR}
\State $p_+ \leftarrow \min(p; \Phi_\mathrm{CAR}(p\Gamma_i ) \geq \epsilon)$ \Comment{Calculate the minimum power to satisfy the QoS constraint.}
\State $p_i^{t+1} \leftarrow \min (\max(p^*,p_+),P_{\max})$ \Comment{Choose the optimal power for CAR if less than $P_{\max}$ and more than $p_+$.}
\Else
\State $p_i^{t+1} \leftarrow \min (p^*,P_{\max})$ \Comment{Choose the optimal power for AAR if less than $P_{\max}$.}
\EndIf

\EndFor

$\Delta \leftarrow \max_i( |p_i^{t+1} - p_i^{t}|  )$
\State $t \leftarrow t+1$

    \EndWhile
\end{algorithmic}
\end{algorithm}

Several comments are in order.
\begin{enumerate}
\item To implement Algorithm 1, a certain coordination degree has to be assumed, which is classical in wireless papers in which best-response-type algorithms are used (which include all the related works on energy-efficient power control such as \cite{meshkatiCDMA06}--\cite{bacci-tsp-2013}, \cite{letreust-tvt-2013}, \cite{betz}, and \cite{zappone-twc-2013}). Concerning the order in which the transmitters update their power level, it has to be noted that this order can be arbitrary and even changed over time and it does not affect the algorithm convergence (see e.g., \cite{bertsekes-1995}). Although it is useful to have transmissions which are time-slotted and for which each transmitter updates its power over a given time-slot, the algorithm can also be implemented when time-slots have different durations. What matters is that only one transmitter updates its power at a given time instance. The corresponding type of best-response dynamic algorithm is referred to asynchronous or sequential best-response dynamic algorithms, which contrast with synchronous or simultaneous best-response algorithms (convergence is generally not guaranteed for the latter).   
\item To update the power levels $m$ times, a duration corresponding to $mN$ time-slots is required.
\item The quantity $\delta>0$ corresponds to the accuracy level wanted for the stopping criteria in terms of convergence to the NE. The initial power values of this algorithm can be arbitrarily chosen, however the choice of $P_{\max}$ is suggested to allow better convergence time as observed from simulations.
\item Algorithm \ref{alg:br} is distributed in the sense that to update his power, a DM only needs to know the SINR corresponding to his chosen power level, i.e., $\mathrm{BR}_{i,\mathrm{X}}(\ul{p}_{-i})$ can be calculated by knowing $\gamma_i$ for some $p_i$. This is typically achieved using just an SINR feedback mechanism and does not require a central entity that provides knowledge of the channel conditions or power levels chosen by the other DMs. Indeed, given the previous SINR of transmitter $i$, the new SINR (on time-slot $t+1$) is estimated as $\gamma_i^{t+1} = \gamma_i^t \frac{p^{t+1}_i}{p^t_i}$. DM $i$ can therefore plug $\gamma_i^{t+1}$ into its payoff expression and compute its best-response by maximizing it w.r.t. to the sole variable $p_i^{t+1}$, since $\gamma_i^t$ and $p_i^t$ are given.
    
\item As a minor remark, note that Algorithm 1 can be initialized arbitrarily. But, as many iterative algorithms, the choice of the initial point can have an impact in terms of convergence speed. Here, the choice of transmitting at full power can be verified (by simulations) to be statistically good in terms of convergence speed for the typical values of the parameters chosen in the numerical section.
\end{enumerate}
\section{Numerical results}
\label{sec:numerical-results}

\subsection{General setup}

Unless explicitly stated otherwise, the following choices and parameters
are assumed for all the simulations provided here:
\begin{itemize}
\item The number of users or transmitters is set to two ($N=2$). This scenario was chosen because the behavior of various metrics like the price of anarchy (PoA) can be easily analyzed in this situation. The case of ``high interference'', as defined below, is also studied to compensate for this choice. In addition, some specific figures also study the case with more interferers.
\item The block success rate function is chosen as in \cite{belmega-tsp-2010}: $f(\gamma_i)=\exp\left[ -\left(\frac{2^{\frac{R}{R_0}}-1}{\gamma_i}\right) \right]$ where $R_0= 1$ MHz is the bandwidth used and the gross data rate is $R=1$ bit/s.
\item When the adaptive arrival rate scenario is considered, it it is assumed that $g(\phi)=\frac{0.1}{\sqrt{\phi}}$.
\item We define the low (resp. high) interference scenario as:
$\mathbb{E}(g_{ii}) = 2.5$ and $\mathbb{E}(g_{ij}) = 0.5$ for $j\neq i$ (resp. $\mathbb{E}(g_{ii}) = 2.5$ and $\mathbb{E}(g_{ij}) = 2$ for $j\neq i$). For some simulations, the channel gains will be assumed to be fixed while for the others it will follow classical block Rayleigh fading. The values indicated will be the instantaneous channel fading when the scenario considered is static and otherwise will indicate the variance.
\item The noise level is set to $\sigma^2 = 1$ mW; the maximum power $P_{\max}=1000$ mW; buffer size of $K=10$; $\epsilon=1$ (packet loss constraint) and the fixed power consumption $b=1000$ mW.
\item To measure the global efficiency of the interference network with respect to the centralized solution, we use the price of anarchy. \cite{Papadimitriou} gave a definition of PoA where the optimal situation corresponds to a PoA which equals $1$, while other situations correspond to a PoA$>1$:
\begin{equation}
\forall \mathrm{X} \in \{\mathrm{CAR}, \mathrm{AAR}\}, \
\mathrm{PoA}_{\mathrm{X}} =  \frac{ \displaystyle{\max_{\ul{p}}} \sum_i u_{i,\mathrm{X}}(\ul{p}) }{\displaystyle{\sum_i} u_{i,\mathrm{X}}(\ul{p}^{\mathrm{NE}})  }.
\end{equation}
\end{itemize}

\subsection{About the considered EE performance metric}

Here we assume a single-user scenario i.e., $N=1$, a fixed channel gain (namely $g_{11} = 2.5$), and the arrival rate to be fixed (CAR scenario). Fig. \ref{fig:EESU} depicts the EE (\ref{eq:eta}) as a function of the chosen radiated power for different values of the fixed consumption cost $b$ and packet arrival rate $q$. First, the figure illustrates what has already been proved through Prop. 3.2 namely, EE is quasi-concave w.r.t. the radiated power. Second, we fix $q$ to one and assess the influence of $b$. As $b$ increases from $0$ to $4000$ mW, the curve becomes less peaky. In fact, if $b$ becomes very high, EE tends to merely becomes a packet success rate function. This means that power control becomes irrelevant since it merely boils down to transmitting at maximum power whatever the channel conditions. Now we fix $b$ to $1000$ mW. By moving from the arrival rate of $q=1$ (framework of \cite{betz}) to $q=0.6$ (with a buffer size of $K=10$), it is seen that the EE curve is quite significantly changed and the optimal radiated power changes from $460$ mW to $320$ mW. In the next section, the gain in terms of radiated power brought by the cross-layer approach is quantified in a more general scenario.

 \begin{figure}[h]
      \begin{center}
        \includegraphics[width=90mm]{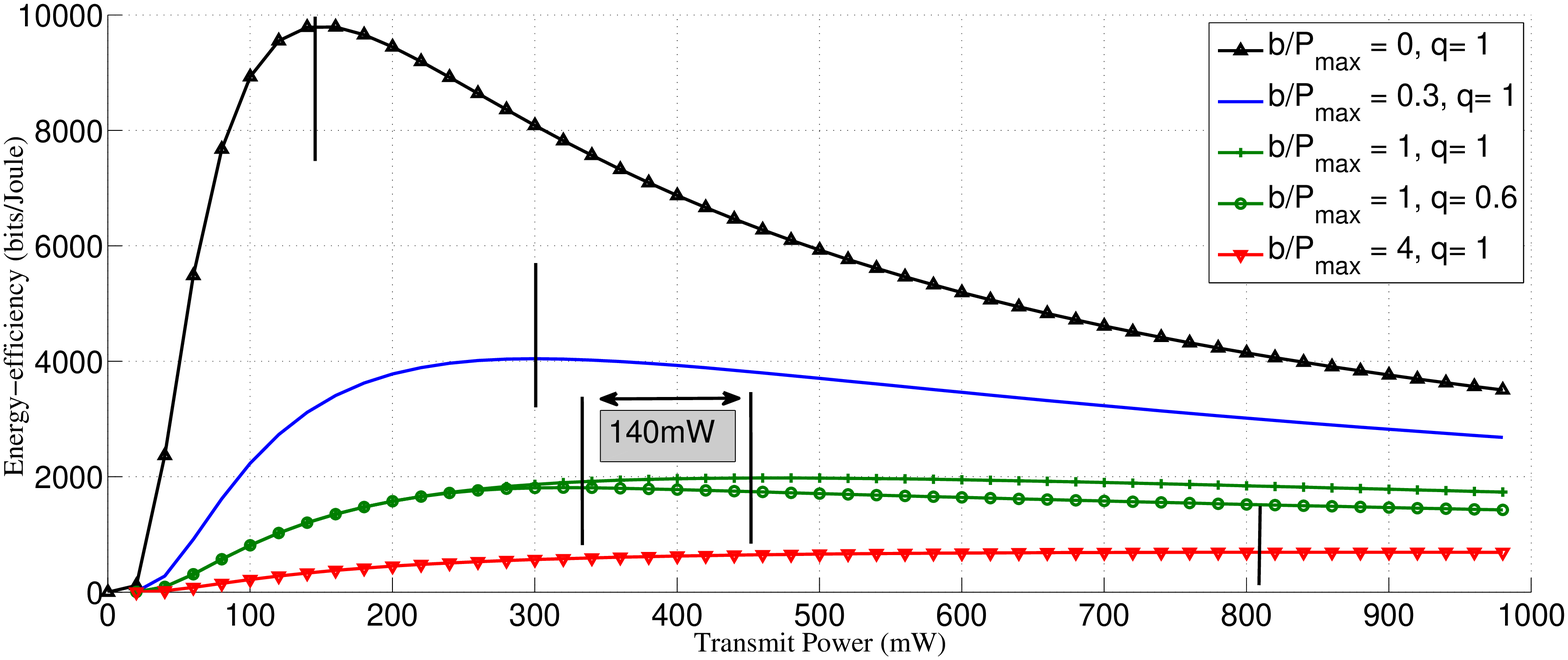}
    \end{center}
    \caption{CAR: EE against $p_1$, i.e., the energy efficiency as a function of the transmit power for various values of the constant power ($b$) and packet arrival rate ($q_{CAR}$).}
    \label{fig:EESU}
\end{figure}

\subsection{Influence of the packet arrival rate in the CAR scenario}

Here we assume the low interference scenario. For $K=10$, Fig. \ref{fig:gainudp} represents the gains in dB in terms of radiated power which is brought by the proposed cross-layer approach (after convergence of the proposed distributed power control algorithm) w.r.t. the conventional approach in which it is (implicitly) assumed that $q \rightarrow 1$ \cite{betz}. The gain is therefore defined by $10\log_{10} \left( \ds{\frac{p_{i}^{\mathrm{NE}}[q \to 1]}{p_{i}^{\mathrm{NE}}[q]}}  \right)$, for a given $i \in \{1,2,3\}$, say $i=1$ (the gain is the same for the different transmitters since the average channel gains are identical). The gain is represented as a function of the packet arrival rate. It is seen that, for different numbers of transmitter-receiver pairs ($N=2$ or $N=3$) and a raw packet error rate of $\epsilon =0.1$ (by raw it is meant before re-transmission), the gain is significant if the arrival rate is typically less than $0.5$. Gains as high as $10$ dB (with $N-1=2$ interfering users on the same band) or $30$ dB (with $N-1=1$ interfering user on the same band). If the raw QoS constraint is relaxed ($\epsilon=1$), quite similar observations can be made. These gains are not in terms of energy consumed by the whole transmit device but they mean that transmitters use much less radiated power and therefore create much less interference, while reaching the same QoS.

 \begin{figure}[h]
      \begin{center}
        \includegraphics[width=90mm]{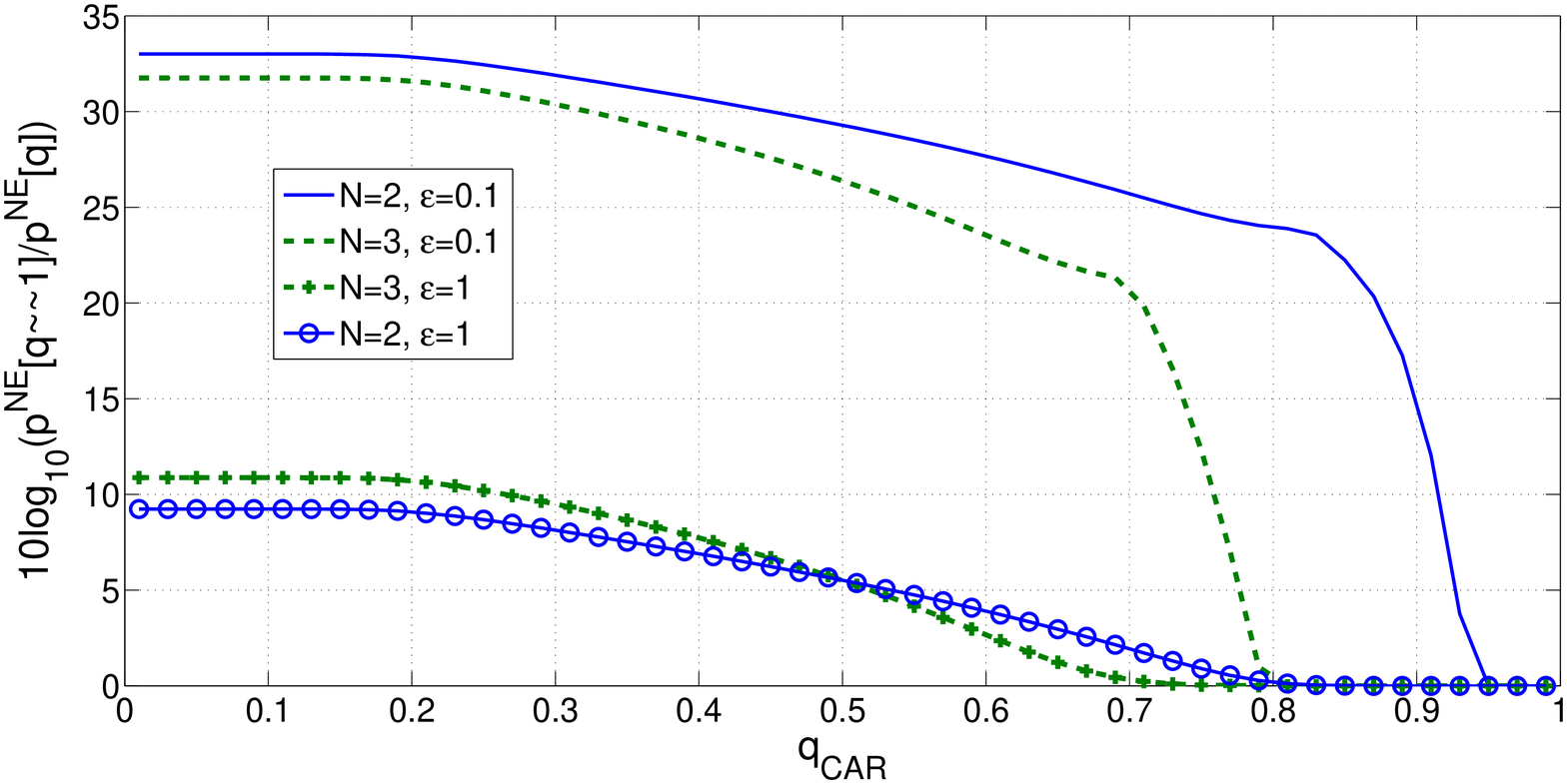}
    \end{center}
    \caption{CAR: $10\log_{10} \left( \ds{\frac{ p_{1}^{\mathrm{NE}}[q \to 1]} {p_{1}^{\mathrm{NE}}[q]}}  \right)$ against $q$, i.e., the ratio of equilibrium power levels in the cross-layer case to the case where the buffer is ignored and arrival rate is one. Interestingly, our cross-layer approach does not only allow the EE to be maximized but also allows significant gains in terms of radiated power. The transmit power for the cross-layer approach is always lower than for the purely physical layer approach, and this difference is more prominent when a packet loss constraint is imposed.}
    \label{fig:gainudp}
\end{figure}

\begin{figure}[h]
    \begin{center}

  \includegraphics[width=90mm]{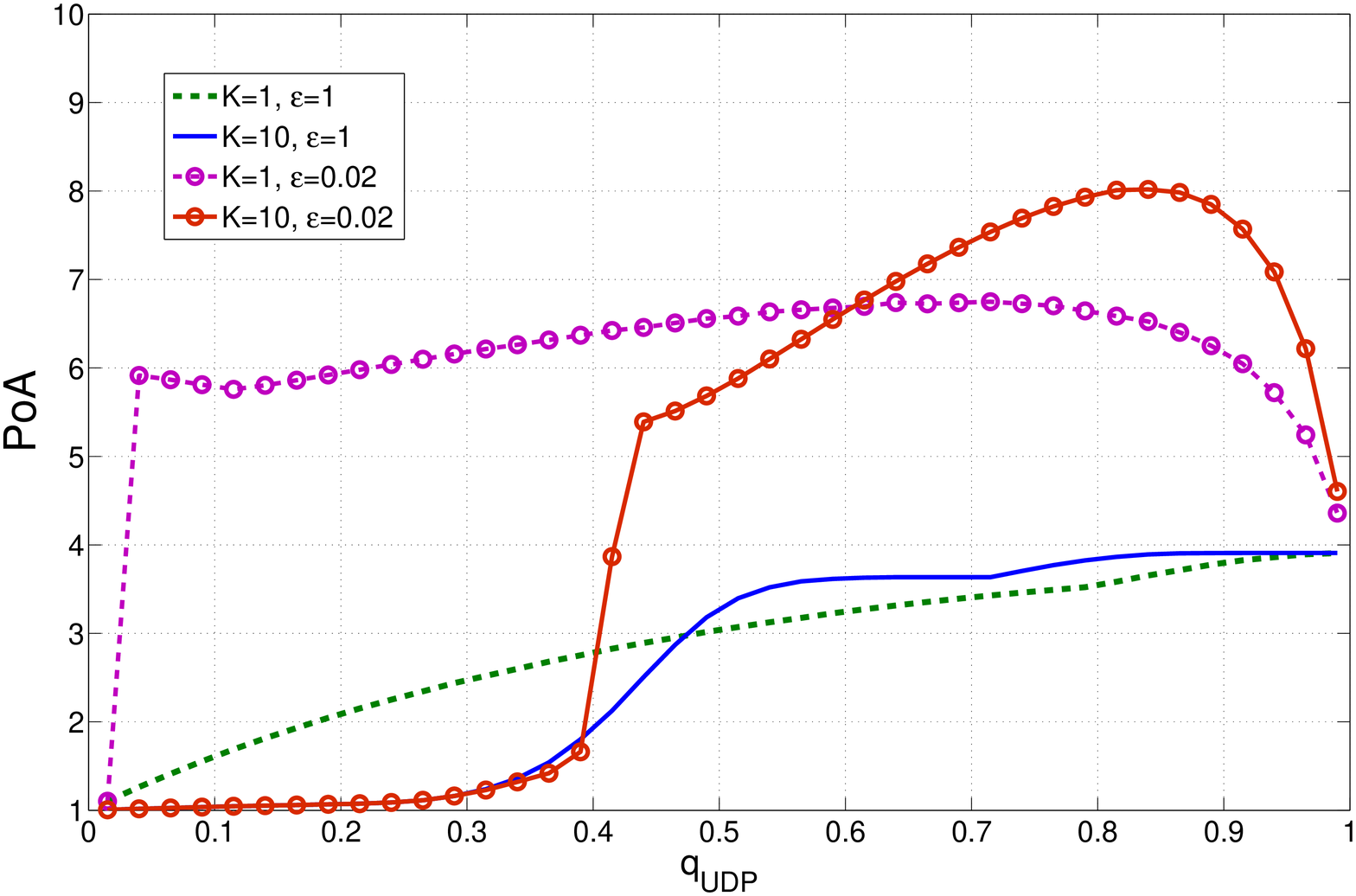}
    \end{center}
    \caption{CAR: Here, a given realization is assumed for the channel. In contrast with existing works on energy-efficient power control which assume $q\rightarrow 1$ and therefore always obtain a high value for the PoA, it is seen here that low values are actually reachable when the packet rate is sufficiently small. With a large enough buffer size ($K$), even for $\epsilon=0.02$, the NE is close to centralized solution if the right $q$ is used.}
    \label{fig:PoAq}
\end{figure}

\begin{figure}[h]
    \begin{center}
  \includegraphics[width=90mm,height=40mm]{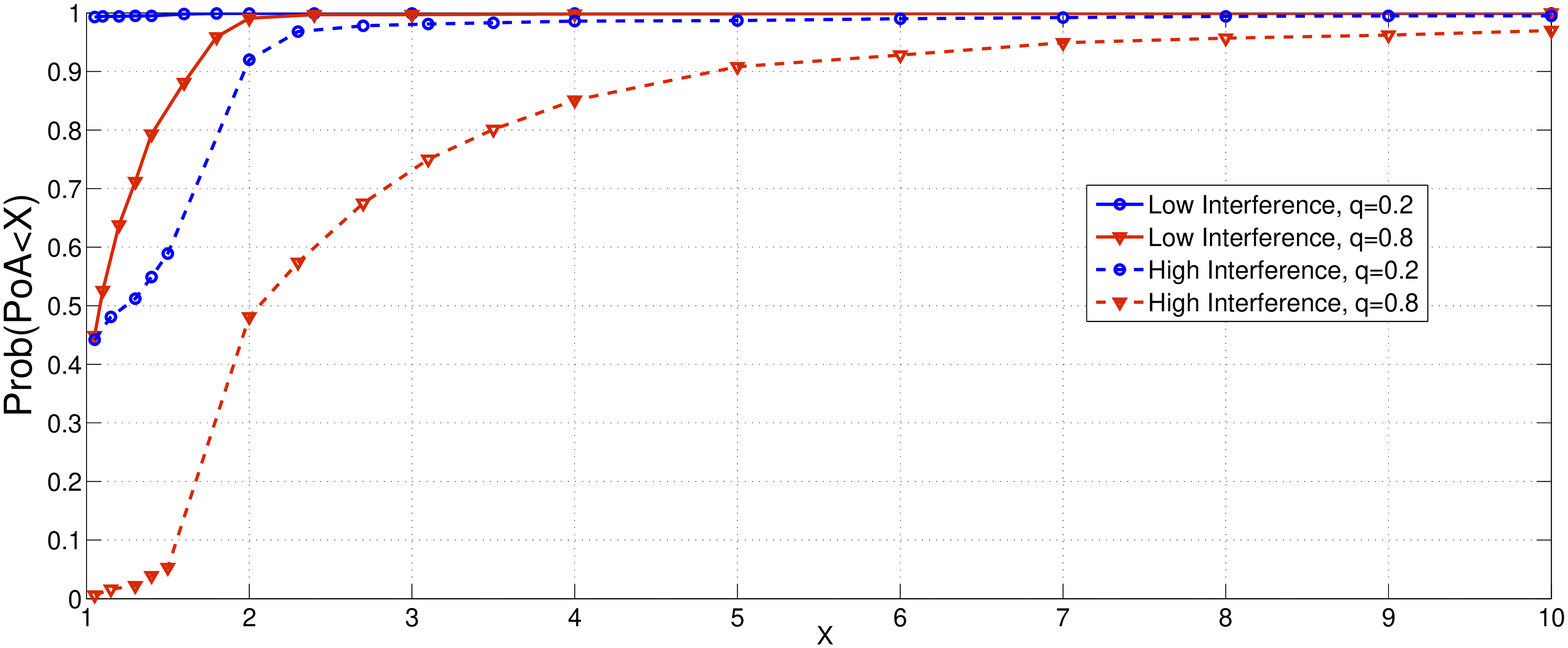}
    \end{center}
    \caption{Contrarily to Fig. \ref{fig:PoAq} which assumes a given channel realization, this figure is obtained by averaging over channel realizations. The CDF of the PoA provides information about how often the price of having a  distributed system is low or high. In this figure as well, we see that even in the high interference regime, a small arrival rate can lead to an efficient equilibrium more often.}
    \label{fig:PoAq2}
\end{figure}

\begin{figure}[h]

 \centering
   \includegraphics[width=90mm]{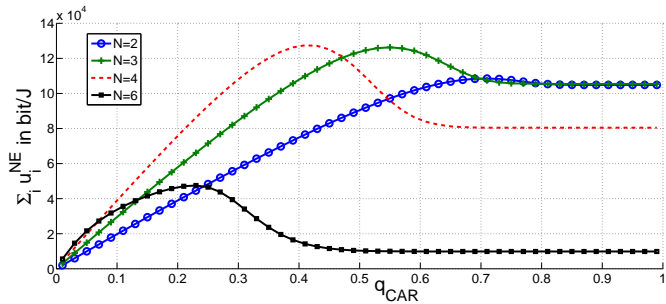}

    \caption{CAR: Remarkably, this figure shows that a communication system can be optimized in terms of used traffic or service. Indeed, there exists an optimal packet rate at which the network EE is maximized. As the number of users increases, the $q$ that corresponds to the best equilibrium, in terms of sum-payoff, decreases. Note that this plot is for the high interference case with $\epsilon=0.02$ resulting in a low equilibrium payoff for $N=8$.}
    \label{fig:NEq}
\end{figure}

In the low interference static channel scenario, Fig. \ref{fig:PoAq} depicts the PoA or price of having a distributed network versus the packet arrival rate for different buffer sizes ($K=1$ and $K=10$) and a raw packet error rate of $\epsilon =0.02$. In contrast with existing works on EE, the PoA can be small in energy-efficient interference networks. This occurs when the arrival rate is typically less than $0.4$ and for a reasonably large buffer size; $K=10$ is in fact quite small while $K=1$ is the minimum buffer size possible and corresponds a very extreme case. The jump observed in the figure around $q=0.4$ at low interference and $q=0.5$ at high interference. This occurs when $q \geq f(\gamma^{NE})$. This jump in the PoA occurs when the value of $q$ crosses this threshold, as the equilibrium power control policy before the jump corresponds to a power control policy closer to the one seen in \cite{goodman-pc-2000}, while after the jump, the equilibrium is closer to the one in \cite{betz}. It is therefore worth noting that, under some realistic conditions, a distributed interference management policies can perform as well as a centralized one. To our knowledge, this observation has not been made before in the literature originating from \cite{goodman-pc-2000} because all the corresponding works assume that the transmitter has always packets to send while this is not the case in many real scenarios (download speeds are often limited by server speeds).

Since our observations regarding the PoA might be thought to be related to the specific realization of the channel, we now provide numerical results which have been obtained by averaging over channel realizations. Fig. \ref{fig:PoAq2} shows the cumulative distribution function (CDF) of the PoA for four parameter settings: Low interference scenario and $q_{\mathrm{CAR}}=0.2$; Low interference scenario and $q_{\mathrm{CAR}} =0.8$; High interference scenario and  $q_{\mathrm{CAR}}=0.2$; and finally, high interference scenario and  $q_{\mathrm{CAR}}=0.8$. This figure confirms that the loss on optimality induced by decentralization is rather small if transmissions are sporadic and interference is not severe.

Now, Fig. \ref{fig:NEq} represents the network sum-payoff, which is an absolute performance measure. In the high interference scenario, for $K=10$, a raw QoS of $\epsilon = 0.02$, the figure depicts the sum-payoff versus $q$ for different numbers of transmitter-receiver pairs ($N \in \{2,3,4,6\}$). This figure illustrates that the sum-payoff at the NE is maximized at a particular $q$ which is seen to decrease with the number of transmitters. This can be intuitively understood, as if the packet arrival rate is reduced, it is possible for more transmitters to experience the same QoS and transmit at a lower power. On the other hand, a very small $q$ implies that the network resources are not being sufficiently exploited, resulting in low efficiency.

\subsection{Gains in terms of energy brought by the cross-layer approach w.r.t. the state-of-the art}

To our knowledge, existing works in the literature originating from \cite{goodman-pc-2000} do not interpret EE maximization as energy minimization. As explained in the paper, both problems are in fact equivalent in communications systems where re-transmissions are allowed. We exploit this interpretation here to go further than just assessing the gains in terms of EE as done classically. Indeed, we assess the gain in terms of energy or average total power brought by the proposed cross-layer approach over the closest state-of-the art solution which is given in \cite{betz} (the latter is obtained by assuming $q\rightarrow 1$ whatever the actual value of $q$). For $q=0.5$ and $q=0.3$, Fig. \ref{fig:Evsb} shows that it is possible to have improvements in terms of energy consumed by the device and not just EE. This (relative) gain can be as high as $28\%$ for $q=0.5$ and $42\%$ for $q=0.3$ in the setting under consideration. Interestingly, this gain can be obtained under the same information assumption as \cite{betz} namely, only individual SINR feedback is needed to implement the power control algorithm which provides the NE performance (after convergence). Note that in this case, $q=1$ offers no gain as the situation is identical to that treated in \cite{betz} while $q \to 0$ would offer maximum gain.

 \begin{figure}[h]
      \begin{center}
        \includegraphics[width=90mm]{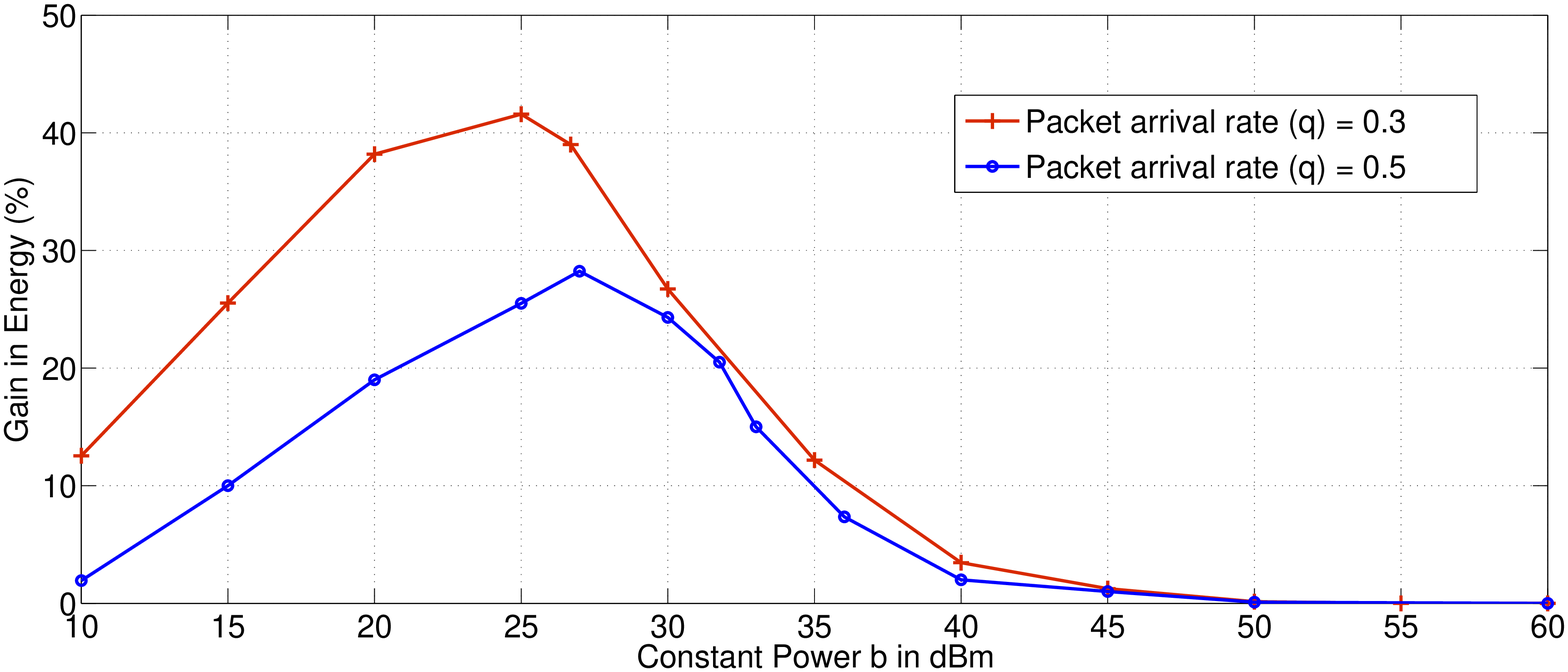}
    \end{center}
    \caption{CAR: Plotting the energy consumed against $b$ with $q=0.6$ and $q=0.3$. We compare the performance of our proposed algorithm against using the best-response dynamics algorithm from \cite{betz} where the presence of the queue is ignored.}
    \label{fig:Evsb}
\end{figure}
 \begin{figure}[h]
      \begin{center}
        \includegraphics[width=90mm]{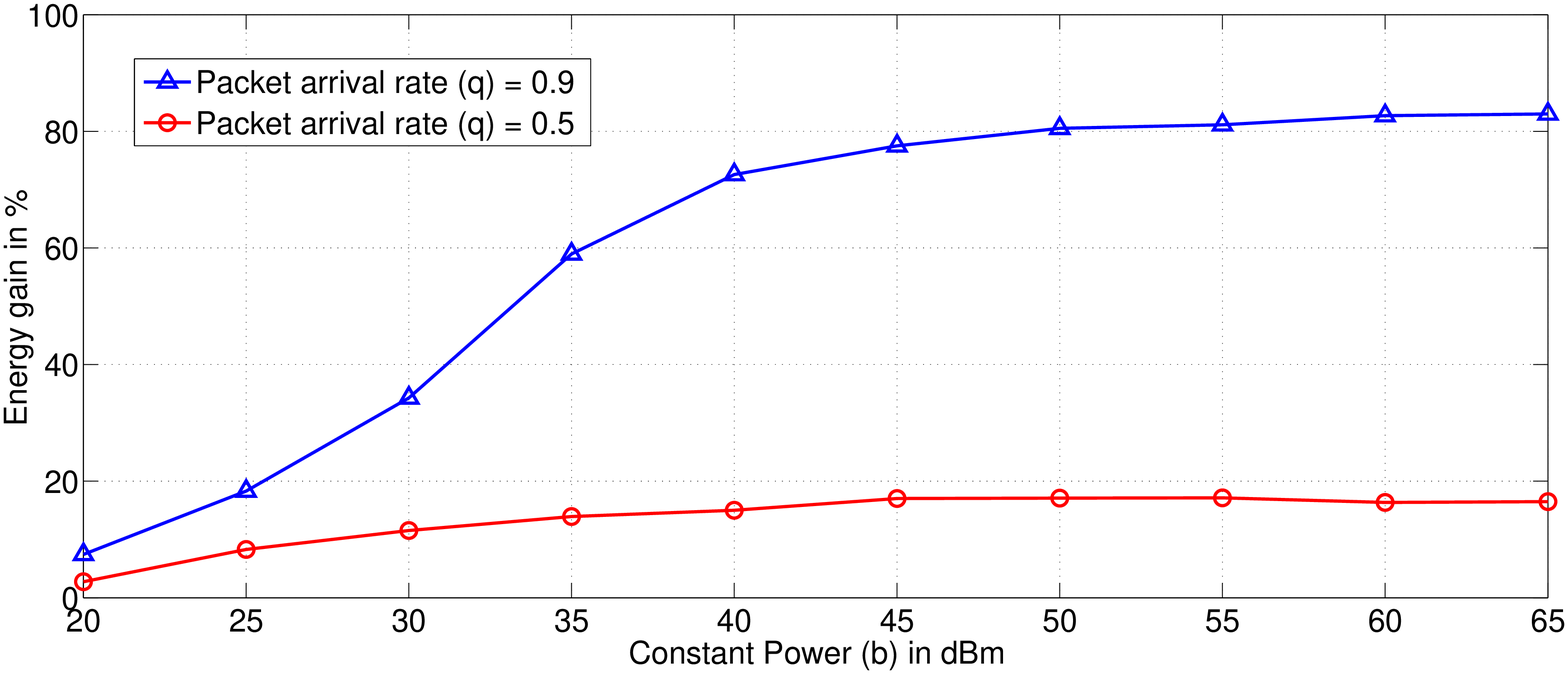}
    \end{center}
    \caption{CAR: Plotting the energy consumed against $b$ with $q=0.5$ and $q=0.9$. We compare the performance of our proposed algorithm against a scheme that just minimizes the transmit power such that the SINR $\geq 25$ dB. We show that our proposed algorithm satisfies this constraint and still consumes less energy.}
    \label{fig:Evsb2}
\end{figure}

As a second comparison in terms of energy, we compared the energy consumed by a transmitter when optimizing (\ref{eq:eta}) with what would obtained by just minimizing the radiated power under an SINR constraint, which is a classical approach. Fig. \ref{fig:Evsb2} corresponds to the relative gain in terms of saved energy as a function of the fixed consumption cost $b$, for $q=0.5$ and $q=0.9$, $R= 8$ Mbps and an SINR target of 25 dB for both approaches in the single user case (interference can make achieving such a target impossible). It is seen that an energy gain of up to $80\%$ can be achieved for sufficiently high values of $b$, which is a quite significant gain and can be easily attained in practice (e.g., maximum radiated power for femto base stations is of the order of one watt while the fixed consumption cost is typically of about a few watts). Note that the gain observed here is maximum when $q=1$ as the highest transmit power is used in this case.

\begin{figure}[h]
    \begin{center}
 \includegraphics[width=90mm]{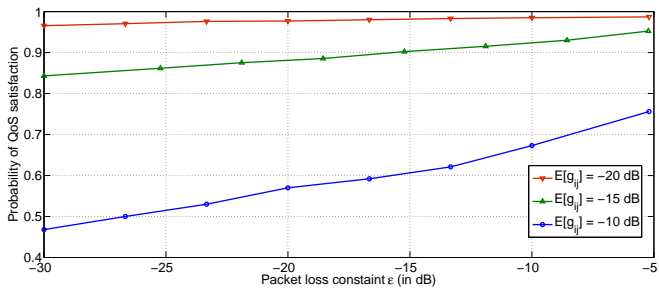}
    \end{center}
    \caption{Probability of meeting the QoS constraint on the packet loss is plotted for $N=3$ users and arrival rate $q=0.5$.}
   \label{fig:qossat}
\end{figure}

\subsection{QoS constraint feasibility}

As the QoS aspect is usually not addressed in the literature of EE power control in the sense of \cite{goodman-pc-2000}, we provide here a figure which shows to what extent the QoS can be met in a typical interference networks where a few users interfere (on a given channel). Fig. 8 represents, for $N=3$ interfering users, the probability that the QoS is met for different values of packet success rate (for instance $-10$ dB represents $\epsilon=10^{-1}$). The average power of the direct links $E(g_ii)$ is set to $0$ dB while different interference levels are considered: $\{-20,-15,-10\}\mathrm{dB}$. This figure clearly indicates the probability with which a terminal will maximize energy-efficiency.

\subsection{Influence of the packet buffer size in the AAR scenario}


So far, we have been assuming the CAR scenario. In particular, this has allowed us to study in detail the influence of the parameter $q$. But, for AAR $q$ is not fixed and varies with the SINR. Fig. \ref{fig:NEtcp} represents, for different numbers of transmitters ($N \in \{2,3,8\}$), the network sum-payoff versus the buffer size for a static channel. The influence of interference (e.g., inter-cell interference) on global energy-efficiency clearly appears. As an important comment, as this simulation shows and many other simulations confirmed this observation (including all simulations assuming CAR instead of AAR), when the buffer size is greater than $10$ typically, the asymptotic regime in terms of buffer size can be assumed to be approximately reached. In practice, this means that, when K is large enough, power control policies might be approximated by implementing the power control policies obtained by assuming $K\rightarrow +\infty$, which corresponds to switching between Cases 1 and 2 (in Sec. \ref{sec:properties}), depending on the current SINR.

\begin{figure}[h]
    \begin{center}
 \includegraphics[width=90mm]{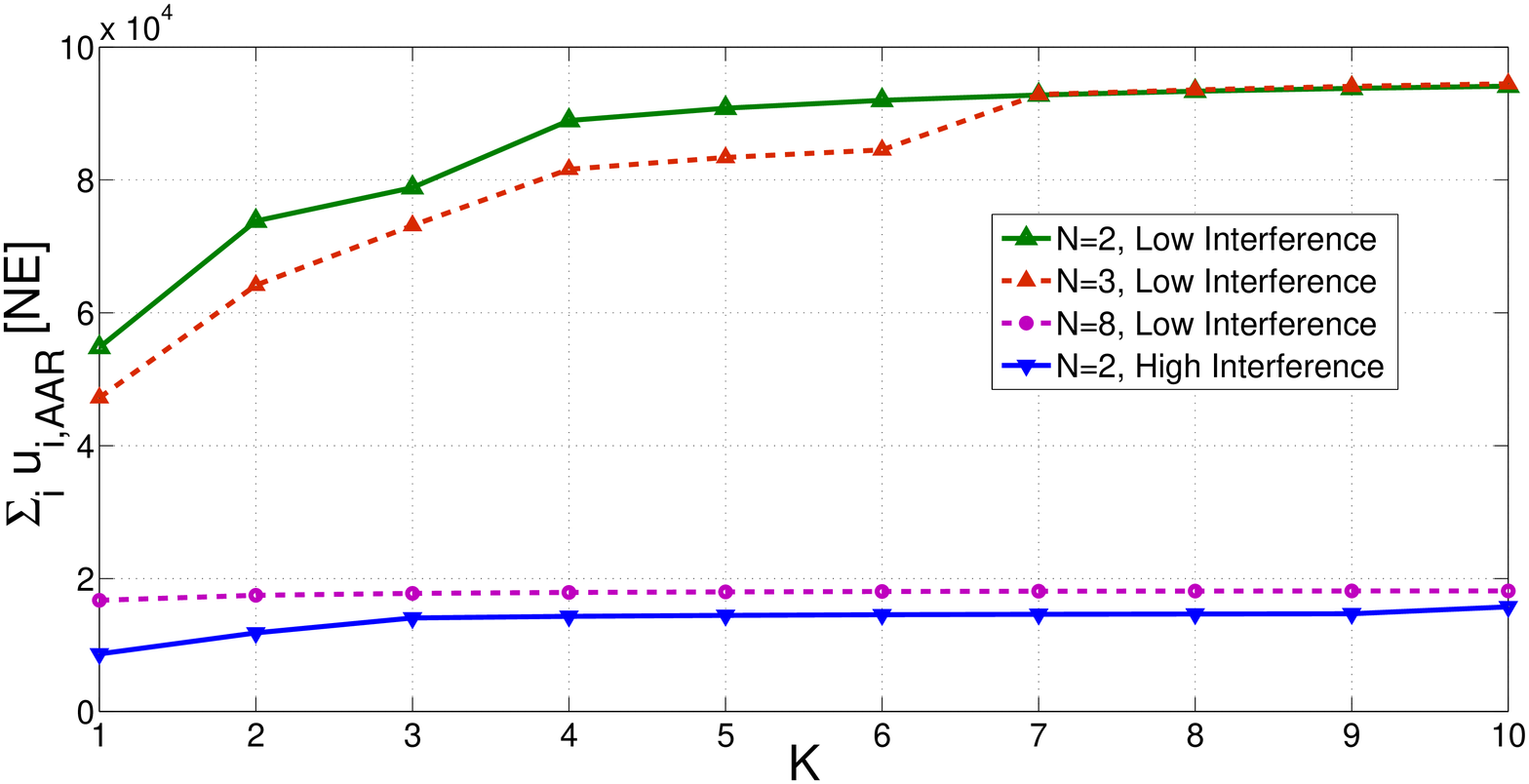}
    \end{center}
    \caption{AAR: We observe that the AAR sum-payoff at the NE is sensitive to the interference level, as seen from the large difference between the two user low and high interference case. With a low interference level, $N=8$ has a higher sum-payoff than for $N=2$ with a high interference level.}
    \label{fig:NEtcp}
\end{figure}

\begin{figure}[h]
    \begin{center}
 \includegraphics[width=90mm]{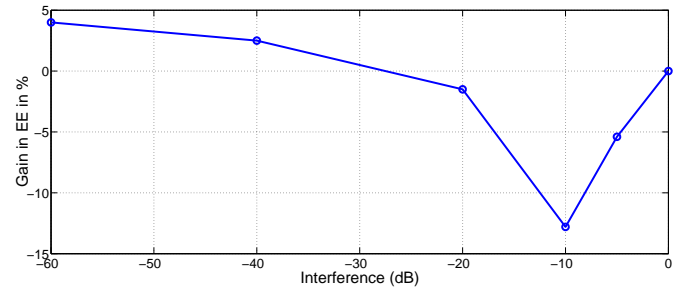}
    \end{center}
    \caption{AAR: Here, we plot the percentage gain in EE v.s $g_{i,j}, i \neq j$ (keeping $g_{i,i}=1$), where the gain is calculated by comparing the EE achieved using the proposed AAR algorithm to the EE at the NE achieved by using the algorithm ignoring the packet level. We observe that in the very low interference regime, the proposed scheme outperforms the other algorithm. However in the low-medium interference region, the NE is inefficient with a high PoA and this results in poor performance.}
    \label{fig:NEtcpvsbetz}
\end{figure}

Fig. \ref{fig:NEtcpvsbetz} studies the average gain in EE (averaged over the channel fading) when compared to that of a power control algorithm ignoring the packet level versus the interference. Here we see that when the interference is very low, the NE of the proposed scheme performs better than an algorithm that ignores the packet level. However, when under interference, the strategy under the AAR scheme would be to use a very high power as EE is individually optimized when the AAR achieves a higher packet rate. This results in a sub-optimal NE as seen in the figure when the interference is in the $[-25,0]$ dB range. This effects indicates that the cross-layer approach might induce some performance loss w.r.t. the classical approach. The authors wanted to emphasize this negative but quite surprising result since it indicates that in distributed networks, refining the modeling aspect can sometimes induce a performance loss; this result can be related to other known paradoxes in distributed networks such as the Braess paradox \cite{Braess-69}.

\section{Conclusion}
\label{sec:conclusion}

Compared to the closest related works, the work reported in this paper possesses three salient features: The (possible) existence of packet buffer with finite size is taken into account; The total power consumed by the transmitter is considered; The proposed formulation considers the QoS. Remarkably, even though the derived energy-efficiency performance metric is seemingly more complex, it possesses all the main properties necessary for designing efficient distributed algorithms. Quite surprisingly, this is not only true when the packet arrival rate is constant (CAR protocol) but also when it is assumed to be adapted as a function of the SINR and the subsequent packet loss through the AAR protocol. One of the consequences of these properties is that the proposed iterative distributed power control algorithm converges towards a unique Nash equilibrium of the power control game associated with both transport protocols.

While the cross-layer generalization of energy-efficient power control is supported by several key analytical results, numerical results strongly support our approach as well. One of the key observations made from simulations is that a distributed power control scheme can perform as well as a centralized solution in some situations; realistic settings under which the PoA is one are clearly identified. Also, it is clearly explained why maximizing EE amounts to minimizing energy in communication systems with re-transmission protocols and this key interpretation is exploited to assess the gain in terms of saved energy brought by the proposed approach.

The proposed approach might be extended in many relevant ways. To address more general wireless scenarios, the most simple extension would be to address the multi-carrier case and also the case of frequency selective channels, these extensions being potentially related. When relevant, receivers might be assumed to implement successive interference cancellation. In order to obtain more efficient equilibrium points (e.g., in the sense of the sum-payoff or a given fairness criterion), it would be of high interest to exploit a more advanced game model such as a stochastic game; this extension is especially relevant if the queue state information has to be exploited. To go further in the direction of having a very realistic wireless network model, a less trivial, but very relevant extension would be to analyze the case of a time-varying number of users. This is definitely both of practical and theoretical interest. Finally, the case of CAR and AAR transmitters simultaneously active in the network has not been studied in this paper. However, our results for quasi-concavity are independent of the protocol used by the other DMs and so the existence of the NE is guaranteed, while uniqueness is left for future extensions.

\section{Acknowledgments}

This work is a joint collaboration between Laboratoire des Signaux et Syst\`{e}mes (L2S) of Sup\'{e}lec, Orange Labs R$\&$D as well as University of Avignon. The work has been done as part of the Operanet 2 European project.

\appendices

\section{Proof of existence of a pure NE in $\mc{G}_{\mathrm{X}}$}
\label{app:qc}

{\bf Proof: } In order to prove the existence of a pure NE in $\mc{G}_{\mathrm{AAR}}$, it is sufficient to show that $\eta_{i,\mathrm{AAR}}(\ul{p}_i)$ is quasi-concave w.r.t $p_i$  as $u_{i,\mathrm{AAR}}(\ul{p}) = \eta_{i, \mathrm{AAR}}(\ul{p})$ and this is a continuous function. However, $u_{i,\mathrm{AAR}}(\ul{p})$ is only semi-continuous and therefore in order prove the existence of a pure NE, proving the quasi-concavity of $\eta_{i,\mathrm{CAR}}(\ul{p}_i)$ (this is done in Appendix \ref{app:qccar}) is only a necessary but not sufficient condition. The additional conditions required to invoke the results in \cite{dasgupta-1986} are proved in the Appendix \ref{app:gepsilon} part of this section.
To prove that $\eta_{i,\mathrm{X}}(\ul{p}_i)$ is quasi-concave w.r.t $p_i$, we consider its reciprocal\\
$\ds{\frac{1}{\eta_{i,\mathrm{X}}(\ul{p}_i )}} =A_i(\ul{p}) + B_{\mathrm{X}}(\gamma_i(\ul{p}))$, where;
$A_i(\ul{p}_i)=\ds{\frac{p_i}{R f(\gamma_i(\ul{p}))}}$, the physical layer factor which depends on the transmit power and the SINR. $B_{\mathrm{X}}(\gamma_i(\ul{p}))=\ds{\frac{b}{R q_{\mathrm{X}}(\gamma_i(\ul{p}))[1 - \Phi_{\mathrm{X}}(\gamma_i(\ul{p}))] }}$, the cross-layer factor which depends on the protocol X and the SINR. Recall that $f(\gamma)$ is increasing and initially convex for $\gamma \in [0,\gamma_+]$ and eventually concave for $\gamma \in [\gamma_+,\infty)$. So, we have that $A_i(\ul{p})$ is decreasing in the interval $\gamma_i(\ul{p}) \in [0,\gamma_+]$ and convex w.r.t $p_i$ for $\gamma_i(\ul{p}) \in [\gamma_+,\infty)$. If a function is continous and differenciable, then it is sufficient to show that it is convex at all local minima/maxima for quasi-convexity. The inverse function we consider is strictly decreasing in the interval $[0,\gamma_+]$ and thus, can not have a maxima/minima in this interval. Hence, once we prove that it is convex in the other interval, we prove quasi-concavity of the original function.
If $B_{\mathrm{X}}(\gamma_i(\ul{p}))$ is a monotonically decreasing function and is convex for $\gamma_i \geq \gamma_{+}$, then we have $\ds{\frac{1}{A_i(\ul{p})+B_{\mathrm{X}}(\gamma_i(\ul{p}))}}$ quasi-concave w.r.t $p_i$ \cite{boyd,rodriguez}.
From (\ref{eq:gamma}), $\ds{\frac{\partial{\gamma_i(\ul{p})} }{\partial p_i}}$ is a constant and in the following sections, to prove that $\ds{\frac{\partial B_{\mathrm{X}}}{\partial p_i}} <0$ and $\ds{\frac{\partial^2 B_{\mathrm{X}}}{\partial p_i^2}}>0$, we just prove that:\\
$$\ds{\frac{\partial B_{\mathrm{X}}}{\partial \gamma_i} }<0, \quad \mbox{and} \quad \ds{\frac{\partial^2 B_{\mathrm{X}}}{\partial \gamma_i^2}}>0.$$

\subsection{Proof of quasi-concavity of $\eta_{i,\mathrm{CAR}}(\ul{p}_i)$}
\label{app:qccar}
In this sub-section we prove that the required conditions for quasi-concavity are satisfied under a CAR scheme. For CAR, $q_{\mathrm{CAR}}(\gamma_i(\ul{p}))=q$ is a constant. Now let us study the derivatives of the function $B_{\mathrm{CAR}}(\gamma_i)$ w.r.t $\gamma_i$;
\begin{equation}
\frac{\partial B_{\mathrm{CAR}}(\gamma_i)}{\partial \gamma_i} = \frac{b}{Rq (1-\Phi_{\mathrm{CAR}}(\gamma_i))^2}\frac{\partial \Phi_{\mathrm{CAR}}(\gamma_i)}{\partial \gamma_i},
\label{eq:d1bp}
\end{equation}
and
\begin{eqnarray}
& \ds{ \frac{\partial^2 B_{\mathrm{CAR}}(\gamma_i)}{\partial \gamma_i^2} =   \frac{b}{Rq (1-\Phi_{\mathrm{CAR}}(\gamma_i))^2}  } \nonumber &\\ &
\ds{ \left(\frac{\partial^2\Phi_{\mathrm{CAR}}(\gamma_i)}{\partial \gamma_i^2} +\left(\frac{\partial \Phi_{\mathrm{CAR}}(\gamma_i) }{\partial p_i} \right)^2 \frac{2}{1-\Phi_{\mathrm{CAR}}(\gamma_i)} \right). }&
\label{eq:d2bp}
\end{eqnarray}
From (\ref{eq:d1bp}), we see that showing $\frac{\partial \Phi_{\mathrm{CAR}}}{\partial \gamma_i} <0$ is sufficient for proving that $\frac{\partial B_{\mathrm{CAR}}(\gamma_i)}{\partial \gamma_i} <0$ as the other terms are always positive.\\
Similarly, $\frac{\partial^2\Phi_{\mathrm{CAR}}}{\partial \gamma_i^2} >0$ is sufficient  for proving that $\frac{\partial^2 B_{\mathrm{CAR}}(p)}{\partial \gamma_i^2} >0$.\\
\begin{equation}
\frac{\partial\Phi_{\mathrm{CAR}}}{\partial \gamma_i} = \left(1-f(\gamma_i)\right)\frac{\partial\Pi_{\mathrm{CAR}}}{\partial \gamma_i} - \Pi_{\mathrm{CAR}} \frac{\partial f(\gamma_i)}{\partial \gamma_i}
\label{eq:d1pip}
\end{equation}
\begin{eqnarray}
&\ds{\frac{\partial^2\Phi_{\mathrm{CAR}}}{\partial \gamma_i^2} = \left(1-f(\gamma_i)\right)\frac{\partial^2\Pi_{\mathrm{CAR}}}{\partial p^2} -} & \nonumber \\ &
\ds{2\frac{\partial\Pi_{\mathrm{CAR}}}{\partial \gamma_i}\frac{\partial f(\gamma_i)}{\partial \gamma_i}-\Pi_{\mathrm{CAR}}\frac{\partial^2 f(\gamma_i)}{\partial \gamma_i^2}.}&
\label{eq:d2pip}
\end{eqnarray}
For  $\frac{\partial\Phi_{\mathrm{CAR}}  }{\partial \gamma_i} <0$, by examining (\ref{eq:d1pip}), we see that showing  $\frac{\partial\Pi_{\mathrm{CAR}}}{\partial \gamma_i} <0$ is sufficient.\\
We have $\omega_{\mathrm{CAR}} = \frac{q}{1-q}\frac{1-f(\gamma_i)}{f(\gamma_i)}$ and so:
\begin{equation}
 \frac{\partial \omega_{\mathrm{CAR}} }{\partial \gamma_i}=\frac{-q}{(1-q)f(\gamma_i)^2}\frac{\partial f(\gamma_i)}{\partial \gamma_i}<0.
\end{equation}
It can be easily verified that $\frac{\partial^2 \omega_{\mathrm{CAR}} }{\partial \gamma_i^2} > 0$ for $\gamma_i \geq \gamma^+$.
Express $\frac{1}{\Pi_{\mathrm{CAR}}}= 1+\frac{1}{\omega_{\mathrm{CAR}}}+...+\frac{1}{\omega_{\mathrm{CAR}}^K}$. Differentiating with respect to $\gamma_i$, we have
\begin{equation}
\frac{\partial\Pi_{\mathrm{CAR}}}{\partial \gamma_i} = \Pi_{\mathrm{CAR}}^2 \frac{\partial\omega_{\mathrm{CAR}}}{\partial \gamma_i}\left(\frac{1}{\omega_{\mathrm{CAR}} ^2}+...+\frac{K}{\omega_{\mathrm{CAR}}^{K+1}}\right) <0.
\label{eq:dpipexp}
\end{equation}
Again, it can be verified that $\frac{\partial^2 \Pi_{\mathrm{CAR}}}{\partial \gamma_i^2} > 0$. Thus, we have:
\begin{equation}
\frac{\partial \Phi_{\mathrm{CAR}} }{\partial \gamma_i} < 0
\label{eq:dphi}
\end{equation}
and
\begin{equation}
\frac{\partial^2 \Phi_{\mathrm{CAR}} }{ \partial \gamma_i^2 } > 0
\label{eq:d2phi}
\end{equation}
Now, following the argument from the start, we have $\eta_{\mathrm{CAR}}(p_i,\ul{p}_{-i})$ to be quasi-concave. Since there exists some power $p_i$ for which $\eta_{i,\mathrm{CAR}}(p_i)$ is maximized, we have proved that there exists a unique $p_i^*$ for which the EE is optimized.
$\blacksquare$
We are able to determine the optimal power $p_i^*$ which maximize the EE function, by solving the following equation:
\begin{eqnarray}
&0 = - \frac{\partial\Phi_{\mathrm{CAR}}}{\partial p_i}\left(b + \frac{p_iq(1-\Phi_{\mathrm{CAR}})}{f(\gamma_i(\ul{p}))}\right) +& \nonumber \\
& (1-\Phi_{\mathrm{CAR}})\left( \frac{\partial\Phi_{\mathrm{CAR}}}{\partial p_i}\frac{p_i}{f(\gamma_i(\ul{p}))} +\frac{\partial (p_i/f(\gamma_i(\ul{p})) )}{\partial p_i}\right).&
\label{eq:etad}
\end{eqnarray}

\subsection{Proof of existence of a pure NE in $\mc{G}_{\mathrm{CAR}}$}
\label{app:gepsilon}

\begin{definition}
Let $f: \mathcal{X} \to \mathcal{R}$ be a function from $\mathcal{X} \subset \mathcal{R}$ to $\mathcal{R}$. $f$ is said to be upper (or lower) semi-continuous at $x_0 \in \mathcal{X}$, if for every $\epsilon>0$, there is a neighborhood $U$ of $x_0$ such that $f(x) \leq f(x_0)+\epsilon,\, \forall x \in U$ (and for lower semi-continuous, $f(x) \geq f(x_0) + \epsilon,\, \forall x \in U$).
\end{definition}

Proof: Here we use the result in \cite{dasgupta-1986} (Corollary of Theorem 2 in \cite{dasgupta-1986}) which states that:
\newtheorem{theory}{Theorem}
\begin{theory}
$\forall i$, $\mathcal{A}_i \subset \mathbb{R}^N$ be non-empty, convex and compact, and let $U_i: \mathcal{A}_i \to \mathbb{R}$ be quasi-concave in $u_i$ and upper semi-continuous. Define $V_i(a_{-i})=\max[U_i(a_i,a_{-i})]$. If $V_i$ is lower semi-continuous in $a_{-i}$ then, the game $(\mc{N},\mathcal{A},U)$ has a pure NE.
\label{th:dasgupta}
\end{theory}

In this section, we prove that $u_{i,\mathrm{CAR}}$ is upper semi-continuous and quasi-concave, and that the newly defined function $V_i(\ul{p}_{-i})=\max[u_{i,\mathrm{CAR}}(p_i,\ul{p}_{-i})]$ is lower semi-continuous. Here, we identify $V_i$ as the payoff of the best-response, i.e., $V_i(\ul{p}_{-i})=u_{i,\mathrm{CAR}}(\mathrm{BR}_{i,\mathrm{CAR}}(\ul{p}_{-i}),\ul{p}_{-i})$.
Studying the specific cases of $\mc{G}_{\mathrm{CAR}}$:
Note that $\theta((\ul{p}))< \eta_{i,\mathrm{CAR}}(\ul{p}), \: \forall \ul{p} \in [0,P_{\max}]^K$ by construction. Define $p_i^{+}(\ul{p}_{-i}): \Phi_{\mathrm{CAR}} (\gamma_i(p_i^+(\ul{p}_{-i}),\ul{p}_{-i}))=\epsilon$ (this is unique as $\Phi_{\mathrm{CAR}}$ is a monotonically decreasing and hence invertible function of $p_i$ as seen from the previous subsection) and $p_i^*(\ul{p}_{-i}): \ds{\frac{\partial \eta_{\mathrm{CAR}}(p_i^*(\ul{p}_{-i}),\ul{p}_{-i})}{\partial p_i} =0}$. There are several cases possible:
\normalfont
\color{black}
\begin{enumerate}
\item $p_i^+(\ul{p}_{-i}) \geq P_{\max}$: Here, $u_{i,\mathrm{CAR}}(\ul{p})$ is a strictly increasing function and maximizes at $P_{\max}$.
\item $p_i^*(\ul{p}_{-i}) \leq p_i^+(\ul{p}_{-i}) < P_{\max}$: Here $u_{i,\mathrm{CAR}}(\ul{p})$ is a strictly increasing function in the interval $p_i=[0,p_i^+(\ul{p}_{-i}))$ and after a point of discontinuity at $p_i^+(\ul{p}_{-i})$, is strictly decreasing in the interval $[p_i^+(\ul{p}_{-i}),P_{\max}]$. So $u_{i,\mathrm{CAR}}$ maximizes at $p_i^+(\ul{p}_{-i})$.
\item $p_i^+(\ul{p}_{-i}) < p_i^*(\ul{p}_{-i})$:  $u_{i,\mathrm{CAR}}(\ul{p})$ is strictly increasing in the interval $[0,p_i^+(\ul{p}_{-i}))$ and after a point of discontinuity at $u_i^+$, is quasi-concave in the interval $[p_i^+(\ul{p}_{-i}),P_{\max}]$. So $u_{i,\mathrm{CAR}}$ maximizes at $p_i^*(\ul{p}_{-i})$.
\end{enumerate}
In all three cases, $u_{i,\mathrm{CAR}}(\ul{p})$ is upper semi-continuous and quasi-concave (See Appendix \ref{app:qccar} for properties of $\eta_{i,\mathrm{CAR}}$). Also, $u_i(p_i^*(\ul{p}_{-i}),\ul{p}_{-i})$ is a continuous function in $\ul{p}_{-i}$ and for small $\ul{p}_{-i}$, $\mathrm{BR}_{i,\mathrm{CAR}}(\ul{p}_{-i})=p_i^*$ . After the point where $p_i^*(\ul{p}_{-i})=p_i^{+}(\ul{p}_{-i})$, as $\ul{p}_{-i}$ increases further,  $\mathrm{BR}_{i,\mathrm{CAR}}(\ul{p}_{-i})=p_i^+$ which is also continuous. And so $\mathrm{BR}_{i,\mathrm{CAR}}(\ul{p}_{-i})$ is in fact continuous and increasing in $\ul{p}_{-i}$.\\
$V_i$ is continuous in the interval $\ul{p}_{-i} \leq \ul{p}_{-i}^+: p_i^+(\ul{p}_{-i}^+)=P_{\max}$ and is given by $\eta_{i,\mathrm{CAR}}$. For $\ul{p}_{-i}>\ul{p}_{-i}^+$, $V_i$ jumps down according to the definition in (\ref{eq:utiludp}) and is thus, lower semi-continuous.  Using Theorem \ref{th:dasgupta}, we have the result that the Game admits a pure NE.
$\blacksquare$

\subsection{The case of AAR}
\label{app:qctcp}
In this sub-section we prove that the required conditions for quasi-concavity are satisifed under an AAR scheme. For AAR, $q_{\mathrm{AAR}}  (\gamma_i(\ul{p}))$ is determined by the (\ref{eq:tcpfpe}). In AAR, from (\ref{eq:tcp}), we know $q=g(\phi)$ and so $\Phi=g^{-1}(q)$ where $g^{-1}$ is the function inverse of $g(.)$ which is assumed to exist, be twice differentiable, strictly decreasing and convex.
And so we have the following equation for $B_{\mathrm{AAR}}(\gamma_i)$:
\begin{equation}
B_{\mathrm{AAR}}(\gamma_i(\ul{p})) =\frac{b } { R( q_{\mathrm{AAR}} (1-  g^{-1}(q_{\mathrm{AAR}} (\gamma_i(\ul{p})))}.
\end{equation}
Now let us study the derivatives of the function $B_{\mathrm{AAR}}$ w.r.t $\gamma_i$ as $\frac{\partial \gamma_i}{\partial p_i} >0$ and is a constant. So the sign of these derivatives do not change even when differentiated w.r.t $p_i$.
\begin{eqnarray}
&\ds{ \frac{-b}{R}  \frac{\partial B_{\mathrm{AAR}}(\gamma_i)}{\partial \gamma_i} =  }&  \\
&
\ds{\frac{q_{\mathrm{AAR}}' (\gamma_i)(1-  g^{-1}(q_{\mathrm{AAR}}) ) - q_{\mathrm{AAR}} (\gamma_i) q_{\mathrm{AAR}}' (\gamma_i )(g^{-1})'(q_{\mathrm{AAR}}) }{[q_{\mathrm{AAR}}(\gamma_i) (1-g^{-1}(q_{\mathrm{AAR}}(\gamma_i) ))]^2} }& \nonumber
\end{eqnarray}
and
\begin{eqnarray}
& \ds{\frac{R}{b} \frac{\partial^2 B_{\mathrm{AAR}}(\gamma_i)}{\partial \gamma_i^2} }= &  \\
&
\ds{\frac{[q_{\mathrm{AAR}}' (\gamma_i)(1-  g^{-1}(q_{\mathrm{AAR}}) ) - q_{\mathrm{AAR}} (\gamma_i) q_{\mathrm{AAR}}' (\gamma_i )(g^{-1})'(q_{\mathrm{AAR}}) ]^2}{[q_{\mathrm{AAR}}(\gamma_i) (1-g^{-1}(q_{\mathrm{AAR}}(\gamma_i) ))]^3}}   & \nonumber \\
& \ds{- \frac{[q_{\mathrm{AAR}}'' (\gamma_i)(1-  g^{-1}(q_{\mathrm{AAR}}) ) -2q_{\mathrm{AAR}}' (\gamma_i )^2 (g^{-1})'(q_{\mathrm{AAR}}) ]}{[q_{\mathrm{AAR}}(\gamma_i) (1-g^{-1}(q_{\mathrm{AAR}}(\gamma_i) ))]^2}+   } & \nonumber \\&
\ds{ \frac{q_{\mathrm{AAR}}(\gamma_i) [q_{\mathrm{AAR}}''(\gamma_i)  (g^{-1})'(q_{\mathrm{AAR}})  +  q_{\mathrm{AAR}}' (\gamma_i )^2 (g^{-1})''(q_{\mathrm{AAR}}) ]}{[q_{\mathrm{AAR}}(\gamma_i) (1-g^{-1}(q_{\mathrm{AAR}}(\gamma_i) ))]^2}   }   &  \nonumber
\end{eqnarray}
From the above expressions, we deduce that the requirements for $B_{\mathrm{AAR}}$ to be decreasing and convex, knowing $(g^{-1})'(q_{\mathrm{AAR}}) \leq 0$ and  $(g^{-1})''(q_{\mathrm{AAR}}) \geq 0$ is that
\begin{equation}
\frac{\partial q_{\mathrm{AAR}} (\gamma_i) }{\partial \gamma_i}  \geq 0.
\label{eq:ineqcondforcon1}
\end{equation}
\begin{equation}
\frac{\partial ^2q_{\mathrm{AAR}} (\gamma_i) }{\partial \gamma_i^2}  \leq 0.
\label{eq:ineqcondforcon}
\end{equation}
Now, we exploit the AAR based fixed point equation:
\begin{equation}\label{eq:tcprel}
g^{-1}(q_{\mathrm{AAR}}(\gamma_i))= \frac{1-f(\gamma_i)}{1+\omega_{\mathrm{AAR}}(\gamma_i)^{-1}+\omega_{\mathrm{AAR}}(\gamma_i)^{-2}+\dots
+\omega_{\mathrm{AAR}}(\gamma_i)^{-K}}.
\end{equation}

Differentiating (\ref{eq:tcprel}) w.r.t $\gamma_i$ once, we get that $\ds{\frac{\partial q_{\mathrm{AAR}} (\gamma_i) }{\partial \gamma_i} } \geq 0 $ and differentiating twice, we get that the inequality (\ref{eq:ineqcondforcon}) is satisfied for $\gamma_i \geq \gamma_{+}$.

Thus we have shown that $\eta_{i,\mathrm{AAR}}$ is quasi-concave w.r.t $p_i$ for the AAR case.
$\blacksquare$

\section{Best-responses for $\mathcal{G}_{\mathrm{X}}$ are standard}
\label{app:standard}

This proof holds for both cases of CAR and AAR. Here, we prove that the best-responses are monotonic and scalable (standard) if $\eta_{i,\mathrm{X}}(\ul{p})$ is quasi-concave w.r.t $p_i$. A function $F(x)$ is standard, if it satisfies the following properties:
\begin{enumerate}
\item  $F(x_1) \geq F(x_2)$, if $x_1 \geq x_2$: Monotonic
\item  $F(\lambda x) \leq \lambda F(x)$, if $\lambda \geq 1$: Scalable
\end{enumerate}

Consider $\ul{P}_{-j} := \lambda \ul{p}_{-j}$, where $\lambda>1$.

$\mathrm{BR}_{j,\mathrm{X}}(\ul{p}_{-j})$ can be calculated by solving for $\gamma^*_j$ in
\begin{equation}
0 = \frac{\partial A(\rho_j^*,p_j)}{\partial p_j} +  \frac{\partial B(\gamma^*_j)}{\partial p_j}
\end{equation}
which can be simplified to
\begin{equation}
0 = \hat{A}(\gamma_j^*) +  C(\ul{p}_{-j})\hat{B}(\gamma_j^*)
\label{eq:derivedf}
\end{equation}
where $\hat{A}(\gamma_j^*)= \ds{\frac{f(\gamma_j^*)-f'(\gamma_j^*)\gamma_j^*}{f^2(\gamma^*_j)}}$,  $C(\ul{p}_{-j})=\ds{\frac{b}{\sigma^2+\sum_{i\neq j}h_ip_i}}$ and $\hat{B}(\gamma_j^*)=\ds{\frac{ \partial B(\gamma_j^*) } {\partial \gamma_j}}$.
As $A$ is convex and $\hat{B}$ negative, (proved in App. \ref{app:qc}), we can conclude that $\gamma^*_j(\ul{P}_{-j}) \leq \gamma^*_j(\ul{p}_{-j}) $.
Thus, $\mathrm{BR}_{j,\mathrm{X}}(\ul{P}_{-j}) \leq  \lambda \mathrm{BR}_{j,\mathrm{X}}(\ul{p}_{-j})$ as $p_j=\gamma_j (\sigma^2 + \sum_{i \neq j} g_{ij} p_i)$.
Therefore the best-responses for the game are scalable.

Now consider $\ul{P}_{-j} \geq \lambda \ul{p}_{-j}$ such that $(\sigma^2 + \sum_{i \neq j} g_{ij} P_i)=\lambda (\sigma^2 + \sum_{i \neq j} g_{ij} p_i)$  , where $\lambda>1$.
Let $\gamma_j^{**}$ (where $\mathrm{BR}_{j,\mathrm{X}}(\dot{\ul{p}_{-j}})=\gamma_j^{**}(\sigma^2 + \sum_{i \neq j} h_i \ul{P}_i)$ is the best-response) satisfy
\begin{equation}
0 = \hat{A}(\gamma_j^{**}) +  \frac{C(\ul{p}_{-j})}{\lambda}\hat{B}(\gamma_j^{**})
\end{equation}
Now replace $\gamma_j^{**}$ by $\frac{\gamma_j^*}{\lambda}$ and we have
\begin{eqnarray}
& \hat{A}(\gamma_j^{*}\lambda^{-1}) +  \frac{C(\ul{p}_{-j})}{\lambda}\hat{B}(\gamma_j^{*}\lambda^{-1}) \leq \\
&\lambda^{-1} \hat{A}(\gamma_j^{*})+  \frac{C(\ul{p}_{-j})}{\lambda^2}\hat{B}(\gamma_j^{*}) \leq  \\
& \frac{\hat{A}(\gamma_j^*) +  C(\ul{p}_{-j})\hat{B}(\gamma_j^*)}{\lambda}  \leq 0.
\end{eqnarray}
The above inequalities are a result of the properties of $\hat{A}$ and $\hat{B}$ given in App. \ref{app:qc}.

Which shows that $\gamma_j^{**}  \geq \frac{\gamma_j^*}{\lambda}$ and thus, $\mathrm{BR}_{j,\mathrm{X}}(\ul{P}_{-j}) \geq \mathrm{BR}_{j,\mathrm{X}}(\ul{p}_{-j})$ and hence the best-responses are monotonic. As all the powers played are positive, the best-response functions satisfy the two requirements and so are standard functions.
$\blacksquare$
\bibliographystyle{plain}

\end{document}